\documentclass[journal]{IEEEtran}

\usepackage{amssymb}
\usepackage{amsmath}
\DeclareMathOperator*{\argmin}{arg\,min}
\usepackage{enumerate}
\usepackage{graphicx}
\usepackage{algorithm}
\usepackage{algorithmic}
\usepackage{xcolor}
\usepackage{makecell}
\usepackage{multirow}
\usepackage{hyperref}
\usepackage{amsthm}
\usepackage[english]{babel}
\usepackage[absolute,overlay]{textpos}
\usepackage{stfloats}
\usepackage{subfig}
\usepackage{pifont}
\newcommand{\cmark}{\ding{51}}
\newcommand{\xmark}{\ding{55}}
\usepackage[linesnumbered,ruled,vlined,algo2e]{algorithm2e}
\usepackage{booktabs} 
\makeatletter
\def\@makecaption#1#2{%
	\footnotesize
	\vskip\abovecaptionskip
	\sbox\@tempboxa{#1. #2}%
	\ifdim \wd\@tempboxa >\hsize
	#1. #2\par
	\else
	\global \@minipagefalse
	\hb@xt@\hsize{\hfil\box\@tempboxa\hfil}%
	\fi
	\vskip\belowcaptionskip}
\makeatother
\newtheorem{theorem}{\textbf{Theorem}}
\newtheorem{lemma}{\textbf{Lemma}}
\newtheorem{remark}{\textbf{Remark}}

\newtheorem{problem}{\textbf{Problem}}
\newtheorem{corollary}{\textbf{Corollary}}
\newtheorem{example}{\textbf{Example}}

\usepackage{cite}
\ifCLASSINFOpdf
  
\else

\fi

\begin{document}

\title{A Dual-AoI-based Approach for Optimal Transmission Scheduling in Wireless Monitoring Systems with Random Data Arrivals}

\author{Yuchong Zhang,~\IEEEmembership{Graduate Student Member,~IEEE,}
	Yi Cao, Xianghui Cao,~\IEEEmembership{Senior Member,~IEEE}
	\thanks{Y. Zhang and X. Cao are with the School of Automation, Southeast University, and Key Laboratory of Measurement and Control of Complex Systems of Engineering, Ministry of Education, Nanjing 210096, China (e-mail: yc\_zhang@seu.edu.cn, xhcao@seu.edu.cn).}
	\thanks{Y. Cao is with the China United Network Communications Co., Ltd., Guangzhou Branch, Guangzhou 510620, China (e-mail: caoyi12@chinaunicom.cn).}
}

\markboth{Journal of \LaTeX\ Class Files}%
{Zhang \MakeLowercase{\textit{et al.}}: A Dual-AoI-based Approach for Optimal Transmission Scheduling in Wireless Monitoring Systems with Random Data Arrivals}

\maketitle

\begin{abstract}
	In Internet of Things (IoTs), the freshness of system status information is crucial for real-time monitoring and decision-making. This paper studies the transmission scheduling problem in wireless monitoring systems, where information freshness---typically quantified by the Age of Information (AoI)---is heavily constrained by limited channel resources and influenced by factors such as the randomness of data arrivals and unreliable wireless channel. Such randomness leads to asynchronous AoI evolution at local sensors and the monitoring center, rendering conventional scheduling policies that rely solely on the monitoring center’s AoI inefficient. To this end, we propose a dual-AoI model that captures asynchronous AoI dynamics and formulate the problem as minimizing a long-term time-average AoI function. We develop a scheduling policy based on Markov decision process (MDP) to solve the problem, and analyze the existence and monotonicity of a deterministic stationary optimal policy. Moreover, we derive a low-complexity scheduling policy which exhibits a channel-state-dependent threshold structure. In addition, we establish a necessary and sufficient condition for the stability of the AoI objective. Simulation results demonstrate that the proposed policy outperforms existing approaches.
\end{abstract}


\IEEEpeerreviewmaketitle

\section{Introduction}\label{sec:introduction}
\IEEEPARstart{I}{n} recent years, the Internet of Things (IoTs) has developed rapidly and has been widely applied in various domains, including smart manufacturing, intelligent transportation and environmental monitoring \cite{11183623, 9445676, 9817083}. In IoT applications, the status information is fundamental, which necessitates real-time monitoring of the underlying physical processes. Utilizing outdated status may lead to incorrect decisions or even catastrophic consequences, thus it is crucial to maintain timely status awareness.

As a network performance metric for quantifying information freshness, the age of information (AoI) has been proposed and attracted increasing attention, which is defined as the time elapsed since the most recently received update packet was generated \cite{sun2022age}. It is worth noting that the loss of information value in IoT systems typically exhibits a nonlinear trend instead of the linear characteristic of AoI. For instance, in remote control systems, the control errors or safety risks induced by the increase of AoI may grow exponentially \cite{9910575}. This motivates the use of an AoI function to map the time metric to information urgency, enabling a more accurate evolution of the penalty associated with AoI updates in the decision-making process. On the other hand, the common ``generate-at-will’’ model assumes that the sensor always possesses fresh data ready for transmission \cite{kadota_scheduling_2018}, but data arrivals are sometimes random, where fresh data packets may be unavailable at the time of scheduling. Specifically, such random data arrival is prevalent in practical network sensing applications and there exist the following three typical scenarios: 1) \textit{Random observation}, in vehicle networks, some sensors generate update packets based on specific events, i.e., vehicle motion causes changes in distance or speed \cite{xie2024scheduling}; 2) \textit{Random energy availability}, in energy-harvesting environmental monitoring, data generated by sensors deployed in fields may be available in a random manner due to the random nature of energy harvesting processes \cite{11079931}; 3) \textit{Random access}, in random access networks, contention-based data transmission may result in random data arrivals \cite{9174254}. Therefore, the sensor-side uncertainty induced by random data arrivals is a non-negligible issue in scheduling policy design.

The randomness in data arrivals may cause conventional deterministic scheduling methods to overlook the fact that no fresh data is available at the sensor-side during scheduling, thereby leading to inefficient decision-making \cite{kadota_minimizing_2021}. As a result, scheduling decisions should consider not only the AoI at the monitoring center (MC) but also the sensor-side data arrival state. In addition, a high AoI at the MC does not necessarily mean that a new update from sensor is worthwhile, since the data packet stored in the buffer may still be stale, which forces the scheduler to make decisions between transmitting an older packet currently or foregoing the scheduling opportunity. Accordingly, both sensor-side data availability and buffered data freshness need to be jointly accounted for.

Furthermore, the randomness of wireless channels introduces a new challenge to this scheduling problem. When the channel quality is poor, a scheduled sensor may not be able to deliver its update packets successfully. Moreover, in practical network transmissions, channel quality may vary dynamically, possessing temporal correlation, and cannot be accurately characterized by the commonly adopted independent and identically distributed (i.i.d.) channel models \cite{yao2022age}. In this case, data arrivals may occur when channel condition is poor, while good transmission windows may appear when no fresh data is available, requiring the scheduler to trade off the transmission of old data packet against possible channel quality reduction upon new arrivals. Due to the randomness of wireless channels and data arrivals, the AoI at the local sensor (AoLI) and that at the receiver in MC (AoRI) evolve in a non-synchronized manner. Hence, it is essential to design a dual-AoI-aware scheduling strategy that jointly captures the effects of data arrival and wireless transmission on monitoring performance.

Motivated by the aforementioned challenges, this paper investigates a transmission scheduling problem for wireless monitoring networks over a shared memory-based channel. To capture information freshness subject to random data arrivals, we consider the asychronous AoLI and AoRI in the transmission process of data packets from multiple local sensors to the MC, based on which we construct an AoI function of monitoring performance and formulate the optimization problem. We introduce a Markov decision process (MDP) to analyze the structure of the optimal scheduling policy and, leveraging this structural insight, design a low-complexity near-optimal scheduling framework. The main contributions of this work are summarized as follows:
\begin{enumerate}
	\item We propose a dual-AoI evolution model, and formulate an optimization scheduling  problem that minimizes the long-term time-average sum AoI functions.
	\item We establish a MDP-based approximate structure of optimal policy through value function decomposition, and reveal that it exhibits a channel state-dependent threshold structure. Then, a structure-informed scheduling algorithm is presented.
	\item We derive a necessary and sufficient condition related to data arrival rate to guarantee monitoring stability over a shared memory-based channel.
\end{enumerate}

The remainder of this paper is organized as follows. In Section 2, related works are reviewed. Section 3 proposes the system model and builds the optimization problem. Section 4 constructs the MDP framework and designs the policy, as well as some theoretical results. The simulation analysis is given in Section 5. Section 6 makes the conclusion.

\textit{Notations:} $\mathbb{N}$, $\mathbb{R}$ and $\mathbb{Z}^+$ denote the sets of all non-negative integers, real numbers and positive integers, respectively. $\mathbb{R}^n$ is the $n$-dimensional Euclidean space. $\Pr(\cdot)$ and $\mathbb{E} [\cdot]$ stand for the probability and expectation of stochastic events, respectively. The transposition, trace and spectral radius of matrix $A$ are denoted by $(A)^T$, $\text{Tr}(A)$ and $\rho(A)$, respectively. $\text{diag}\{ \cdot \}$ denotes the diagonal matrix. $\Vert \alpha \Vert$ is the 1-norm of vector $\alpha$. Inequalities between two vectors are interpreted elementwise, expressed as $\succeq$ or $\preceq$. 

\section{Related Works}
The freshness of status updates delivered to the receiver is paramount for ensuring precise monitoring under the wireless network environment. Consequently, substantial research efforts have been dedicated to quantifying this freshness and optimizing transmission scheduling accordingly. In this section, we first discuss the fundamental developments in AoI minimization and its embedded information urgency, and subsequently investigate the impact of random data arrival models and uncertain wireless channel characteristics on scheduling design.

To characterize the timeliness of system states, AoI has been proposed from the receiver’s perspective as a network-layer metric \cite{sun2022age}. In some early studies, such as \cite{10663703} and \cite{yates2018age}, the time-average AoI of queue systems was analyzed to guide the packet updates. To fully exploit the role of AoI in evaluating the performance of status update systems, the design of AoI-based scheduling has been extensively studied. In \cite{abd2020aoi, de2023unified}, the transmission scheduling mechanisms aimed at AoI optimization were investigated in single-source system scenarios. With hope to extended the AoI-aware scheduling research to multi-source systems, \cite{kadota_scheduling_2018} proposed the near-optimal Max-Weight and Whittle’s index policies. Considering the practical constrains in network environments, \cite{zhu_aoi_2021} introduced AoI minimization to design energy-efficient data collection scheduling policy. In \cite{10269649} and \cite{wang2023minimize}, learning-based methods have also been applied to address AoI-optimizing scheduling problems under environmental uncertainty. Similar method was also applied in \cite{11113352} to design the trajectories of uncrewed aerial vehicles jointly combined with AoI optimization. However, AoI alone is insufficient to capture the information urgency, especially the dynamics of underlying devices \cite{10955408}. In terms of this issue, a metric of AoI reduction was proposed in \cite{pan2022optimizing} to address the multi-source information update problem. In \cite{10778564}, the information urgency was modeled as a function of distortion measure and relative value iteration algorithm was applied to schedule the sampling and transmission. Hence, \cite{sun2017update, hu2021alpha, cao2023optimal} constructed nonlinear AoI penalties to guide the real-time motoring optimization. Similarly, a lightweight multi-sensor scheduler was proposed in \cite{chang_lightweight_2024} based on AoI function-based estimation system model. Unlike the linear evolution of AoI, the dynamics of AoI-based functions are more complex, making traditional AoI optimization methods difficult to adopt directly.

It is worth noting that the data generation model has a significant impact on scheduling design. The aforementioned studies mainly adopt the “generate-at-will’’ model, in which a fresh status update is generated immediately whenever a transmission to the receiver is assigned. As a result, scheduling decisions only need to consider the instantaneous AoI value. This assumption becomes invalid when update packets arrive randomly, making many existing methods inapplicable, thus several recent studies have begun to address this challenge. \cite{hsu_scheduling_2020} developed two scheduling methods for random information arrivals with no buffer based respectively on MDP framework and Whittle’s index. \cite{tang2020scheduling} further incorporated bandwidth constraints for this problem. In \cite{kadota_minimizing_2021}, queueing disciplines were considered to evaluate the impact of random arrivals on AoI. Moreover, \cite{liu_optimizing_2023} modeled the partial information uncertainty caused by random arrivals as a partially observation problem. While these works effectively address source-side uncertainty, they overlook the physical layer characteristics of wireless links, typically assuming i.i.d. channels \cite{pan2022optimizing}. In environments involving networks made up of multi-devices, wireless channels often exhibit temporal correlation or memory \cite{yao2022age}. In \cite{wang2021remote}, the Gilbert-Elliott (GE) model was used to describe the packet loss occurring in transmissions. As reported in \cite{wei_double_2023}, this memory was modeled as a Markov chain and studied the transmission stability problem for single system scenarios. Thus, it remains an open problem that how to account for the effects of random data arrivals and ensure both the timeliness and stability of monitoring coupled with memory-based channels.

The summary of related works are listed in Table \ref{tab1}. In contrast to our work, most existing studies focus primarily on minimizing the AoI value itself, without exploiting the information urgency expression of AoI-based performance functions. Although MDP formulations are widely used, the structural design of scheduling policies remains underexplored. Furthermore, by jointly accounting for stochastic arrivals and dynamic channels, our framework introduces new challenges for policy stability analysis. 

\begin{table*}[htbp]
	\centering
	\caption{Summary and Comparison of Related Works}
	\resizebox{\textwidth}{!}{
		\begin{tabular}{lccccc}
			\toprule
			\multirow{3}{*}{\textbf{References}} & 
			\multirow{3}{*}{\textbf{\makecell{Optimization\\Objective}}} & 
			\multirow{3}{*}{\textbf{Main Methods}} & 
		    \multicolumn{3}{c}{\textbf{Items of Interest}} \\ \cmidrule(lr){4-6}
			 &  &    & \textbf{\makecell{Data Arrival \\ Randomness}} & \textbf{\makecell{Channel Dynamics}} & \textbf{\makecell{Policy Stability}} \\
			\midrule
			\cite{abd2020aoi}             & Time-average AoI & MDP, relative value iteration  & \xmark & \xmark & \xmark \\
			\cite{de2023unified}          & Time-average AoI and energy & MDP, Q-learning  & \xmark & \xmark & \xmark \\
			\cite{wang2023minimize}       & Weighted sum average AoI & Constrained MDP, threshold-based policy  & \xmark & \xmark & \xmark\\
			\cite{11113352}               & Peak AoI & MDP, deep reinforcement learning  & \xmark & \xmark & \xmark\\
			\cite{pan2022optimizing}      & Time-average AoI penalty & Greedy algorithm   & \xmark & \xmark & \xmark \\
			\cite{10778564}               & Time-average AoI penalty & MDP, deep reinforcement learning  & \cmark & \xmark & \xmark \\
			\cite{hu2021alpha}            & Average AoI penalty & Theoretical method   & \xmark & \xmark & \xmark\\
			\cite{cao2023optimal}         & Time-average AoI penalty and energy & MDP, theoretical method   & \xmark & \xmark & \cmark \\
			\cite{chang_lightweight_2024} & Time-average AoI function & MDP, Whittle-index  & \xmark & \xmark & \cmark\\
			\cite{hsu_scheduling_2020}    & Time-average AoI & MDP, relative value iteration   & \cmark & \xmark & \xmark\\			
			\cite{tang2020scheduling}     & Time-average AoI & MDP, relative value iteration   & \cmark  & \cmark & \xmark\\
			\cite{liu_optimizing_2023}    & Time-average AoI & MDP, Max-Weight policy  & \cmark & \xmark & \xmark \\
			\cite{wang2021remote}         & Time-average error and energy & MDP, relative value iteration & \xmark & \cmark & \cmark \\
			\cite{wei_double_2023}        & Time-average error and energy & MDP, relative value iteration  & \xmark & \cmark & \cmark \\
			\cite{yin2020application}     & Time-average AoI-realted value & MDP, deep reinforcement learning  & \xmark & \xmark & \xmark\\
			\cite{11007613}               & AoI penalty & Constrained MDP, relative value iteration  & \cmark & \xmark & \xmark \\
			\cite{wang2024scheduling}     & Time-average AoI & Constrained MDP, Lyapunov optimization  & \cmark & \xmark & \xmark\\
			\textbf{This Work}            & Time-average dual-AoI function & \makecell{MDP, structure analysis, low-complexity theoretical method}  & \cmark & \cmark &\cmark\\
			\bottomrule
		\end{tabular}%
	}
	\label{tab1}
\end{table*}

\section{System Model and Problem Formulation}
As shown in Fig. \ref{fig1}, we consider a typical wireless monitoring networks, where the RMC carries the real-time monitoring through acquiring data transmitted from $N$ local sensors. Each sensor performs the local observation of an  independent physical system, generates observation packets randomly, and transmits them to the RMC over a shared wireless channel. Considering the randomness of data transmission, a buffer is configured to store the freshest data packet, which means that the older data packet would be replaced upon the arrival of a newer packet. Set $i \in \mathcal{N} = \{1,2, \dots, N\}$ as the index of local sensor, and the index of time slot is defined as $k \in \{1,2, \dots, T\}$.

\begin{figure}[ht]
	\renewcommand{\figurename}{Fig.}
	\centering
	\includegraphics[width=0.3\textwidth]{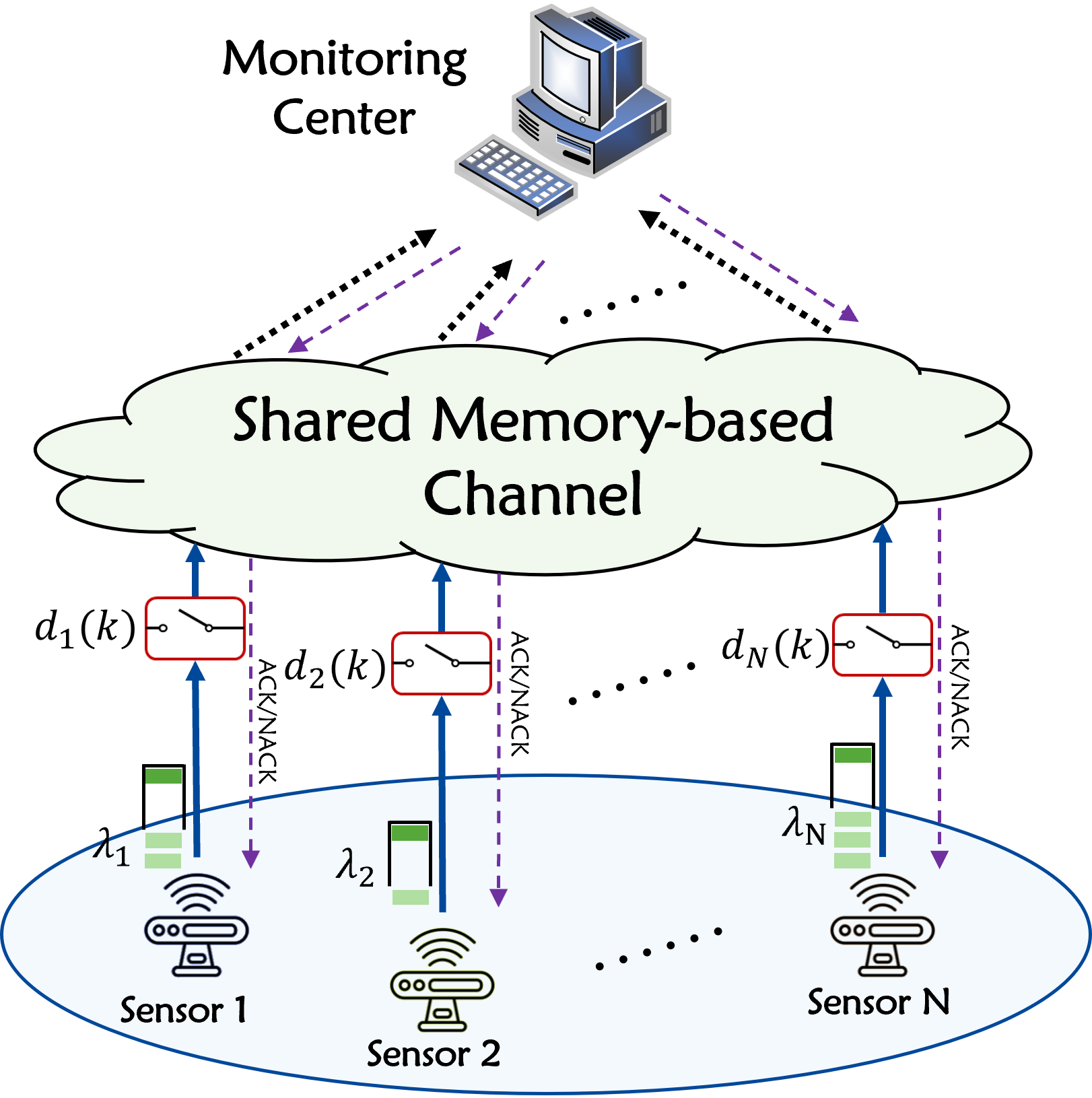}
	\caption{The structure of wireless monitoring.}
	\label{fig1}
\end{figure}

\begin{figure}[ht]
	\renewcommand{\figurename}{Fig.}
	\centering
	\includegraphics[width=0.4\textwidth]{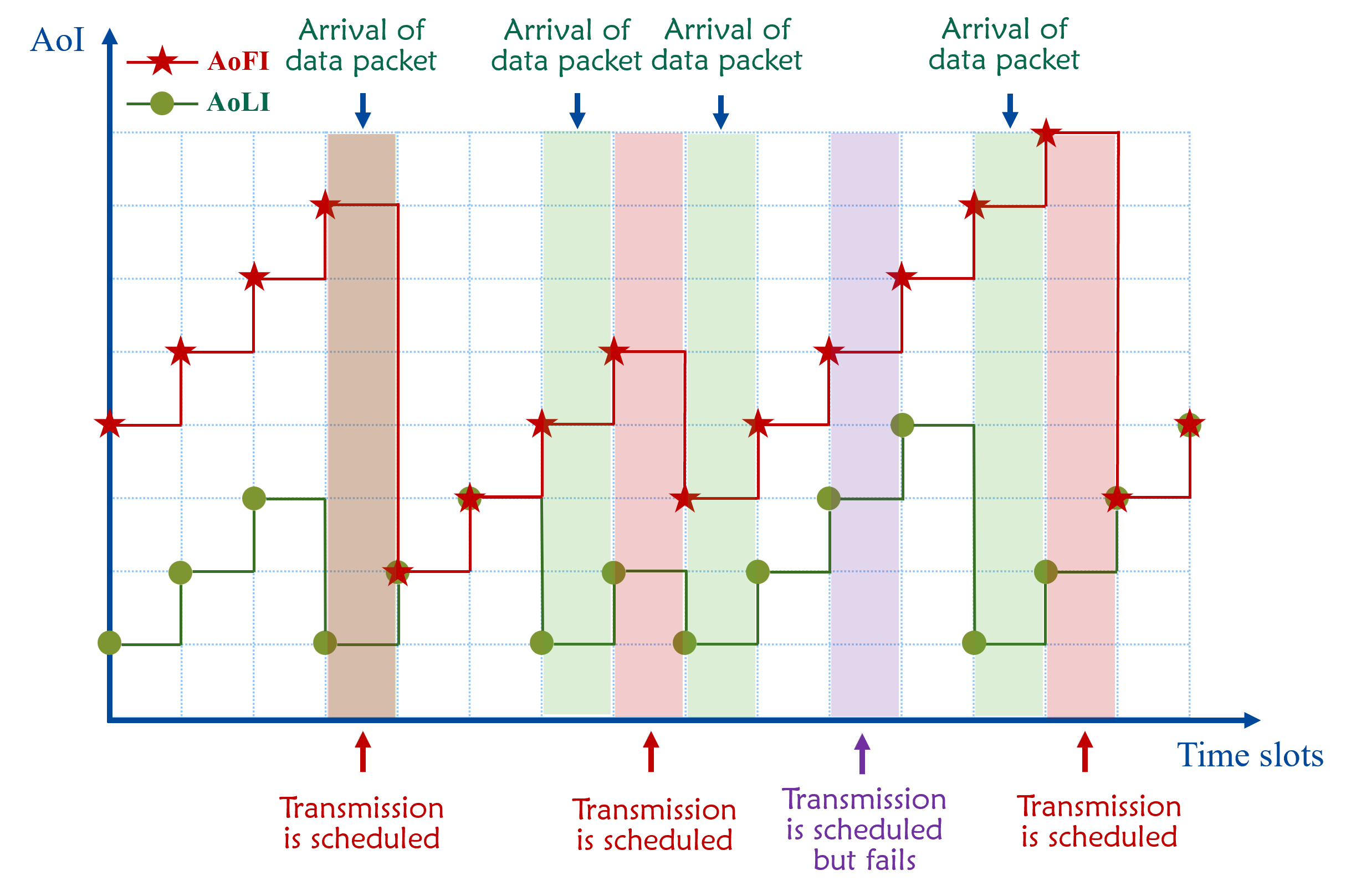}
	\caption{Illustration of AoLI and AoRI.}
	\label{fig2}
\end{figure}

\subsection{Memory-based Channel Model}
In the wireless network environment, obstructions, reflections, or network conditions may cause dynamic variations or fading of transmission channel quality. The static channel models, i.e.,independent Bernoulli distribution, commonly used in existing studies, fail to explain such phenomenon. In this case, the change of channel quality is not completely random but exhibits temporal correlations. To capture this channel memory characteristic, we introduce the GE model, which is widely applied to describe the packet loss in practical wireless networks and confirmed to represent Rayleigh-fading channels accurately \cite{wang2021remote, ayan2023optimal}. The channel fluctuates between "Good'' and "Bad'' states, naturally viewed as a two-state Markov chain, where two states are denoted by $\vartheta(k) = 1$ and $\vartheta(k) = 0$ at time slot $k$. The transition probabilities of channel states are defined as follows:
\begin{align} \label{eq1}
	& \Pr \{\vartheta(k) = 0 | \vartheta(k-1) = 0\} = \kappa_{00}, \nonumber \\
	& \Pr \{\vartheta(k) = 1 | \vartheta(k-1) = 0\} = \kappa_{01}, \nonumber \\
	& \Pr \{\vartheta(k) = 1 | \vartheta(k-1) = 1\} = \kappa_{11}, \nonumber \\
	& \Pr \{\vartheta(k) = 0 | \vartheta(k-1) = 1\} = \kappa_{10},
\end{align}
where $\kappa_{00}, \kappa_{01}, \kappa_{11}, \kappa_{10} \in (0,1)$ are given transition probabilities. Besides, the probabilities of switching different states are computed by $\kappa_{01} = 1 - \kappa_{00}$ and $\kappa_{10} = 1 - \kappa^\vartheta_{11}$. Define $\theta_i(k)$ as an indicator of data reception at time slot $k$ when sensor $i$ is chosen to transmit data, i.e., $\theta_i(k) = 1$ means that data is received by MC, otherwise, $\theta_i(k) = 0$. Then, the successful transmission probability of sensor $i$ is denoted as
\begin{align}
	p^\vartheta_i(k) = \Pr \{ \theta_i(k) = 1 | \vartheta(k), d_i(k) = 1\},
\end{align}
where $d_i(k)$ represents the scheduling decision for sensor $i$. Specifically, when sensor $i$ is scheduled, $d_i(k) = 1$; otherwise, $d_i(k) = 0$. The packet loss probability of the $i$-th sensor is expressed as $1 - p^\vartheta_i(k)$. Moreover, we assume that the ACK/NACK signals are fedback to sensors from the MC with neglecting the delay to inform whether the transmission is successful or not.

Since the bandwidth limitation of the shared channel, it is impossible to transmit data from all sensors in each time slot. Therefore, the following constraint is proposed for the scheduling problem:
\begin{equation}\label{3}
	\sum\limits_{i = 1}^N {{d_i}(k)}  \leq M, \forall k \in \mathbb{N},
\end{equation}
which means that at most $M$ out of $N$ sensors can be selected.

\subsection{AoI Function Regarding Monitoring Performance}
The existing researches mainly focus on the case where updated data packets are directly sent to destination upon generation, which is only appropriate for the sensors with stable transmission. In this paper, we consider a more general scenario of wireless monitoring networks that data packets are randomly available for transmissions. Considering the asynchronous variation in information freshness caused by the random update of data packets, it is hard to conduct the scheduling decisions based solely on the directly observed AoI, because the potential benefit of scheduling cannot be inferred. Hence, we divide the AoI evolution into two stages for analysis: AoLI and AoRI. Here we use a buffer to store the most recently available data packet. As reported in \cite{kadota_minimizing_2021}, the queue mechanism with only single packet can achieve the best freshness under random arrivals. 

Let $\Delta_i^L(k)$ denotes the state of AoLI, measuring the elapsed time after the last update packet of local data arriving at the buffer, and $\gamma^L_i(k)$ as an indicator to show whether updated data packet arrives at buffer for sensor $i$. If there is no update arrival, $\gamma^L_i(k) = 0$, $\Delta_i^L(k)$ will increase by one at the beginning of next time slot, otherwise the old packet is replaced by a fresh one and $\Delta_i^L(k)$ is reset to zero with $\gamma^L_i(k) = 1$. The network-induced events follow a stochastic process, such that the arrival rate is supposed as i.i.d. for each time slot and various sensors, which is modeled as a Bernoulli distribution with probability $\lambda_i$, i.e., $\Pr \{\gamma_i^L = 1\} = \lambda_i, i \in \mathcal{N}$. The AoLI is modeled as follows:
\begin{align}  \label{12}
	{\Delta _i^L}(k) = 
	\begin{cases}
		0, & \text{if} \ \ \gamma^L_i(k) = 1, \\
		{\Delta _i^L}(k-1)+1, & \text{if} \ \ \gamma^L_i(k) = 0.
	\end{cases}
\end{align}

When local data from sensor $i$ is chosen to transmit and succeed, the current data packet of the $i$-th buffer reaches the MC, so that AoRI is updated to $\Delta _i^L(k-1) + 1$. If not, AoRI would increase by one. In view of scheduling decisions and channel qualities, we introduce an indicator $\gamma^R_i(k)$ to represent the real-time situation of data reception for sensor $i$, denoted as
\begin{align} \label{6}
	{\gamma^R_i}(k) = {d_i}(k){\theta _i}(k).
\end{align}
Then, the AoRI dynamics is described as follows:
\begin{align}  \label{13}
	{\Delta _i^R}(k) =
	\begin{cases}
		{\Delta _i^L}(k-1) + 1, & \text{if} \ \ \gamma^R_i(k) = 1,\\
		{\Delta _i^R}(k-1) + 1, & \text{if} \ \ \gamma^R_i(k) = 0.
	\end{cases}
\end{align}
where $\Delta_i^R(k)$ measures the time elapsed since the update packet was generated. The evolution processes of AoLI and AoRI are plotted in Fig. \ref{fig2}. 

In order to reveal the influence of information freshness on monitoring performance, we introduce the AoI function $f(\cdot)$ to measure the information urgency. Depending on system characteristics, various AoI functions have been established in the existing literature. \cite{hu2021alpha} proposed a monotone non-decreasing function in the exponential form to depict the performance of universal nonlinear system. In addition, similar function was developed in \cite{cao2023optimal} for remote control. Then, two typical examples are given as below:
\begin{example} \label{exam2}
	In many real-time network applications, such as multi-user status monitoring in V2X interconnection, online learning relies heavily on fresh data \cite{sun2017update}. To reflect the data urgency, a commonly used AoI-penalty function is defined as
	\begin{equation} \label{3.9}
		f(\Delta_i^R(k)) = e^{r\Delta_i^R(k)} - 1,
	\end{equation}
	where $r > 0$ is a given constant.
\end{example}

\begin{example} \label{exam1}
	In typical remote state estimation systems, sensor networks can perform the local estimation using Kalman filtering. Consider a discrete-time linear time-invariant model as the monitoring object, then the estimation performance is mapped to the following function \cite{chang_lightweight_2024}:
	\begin{align} \label{3.8}
		& f(\Delta _i^R(k))  =  \text{Tr}({h^{{\Delta_i^R}(k)}}({{\bar P}_i})), 
	\end{align} 
	where the function $h(\bar P_i) \triangleq {A_i} \bar P_i A_i^T + \Sigma _i^\omega $, $h^n( \cdot ) = h(h^{n - 1}( \cdot ))$ and $\bar P_i$ is a steady-state after long run. $A_i$ represents the system matrix and $\Sigma _i^\omega$ is the covariance matrix of system process noise. Specially, \eqref{3.8} is easy to expanded as $\text{Tr} \left( A_i^{\Delta_i^R(k)} {{\bar P}_i} (A_i^T)^{\Delta_i^R(k)} + \sum \nolimits_{l=0}^{\Delta_i^R(k)-1} A_i^l  \Sigma_i^\omega (A_i^T)^l\right)$.
\end{example}

\subsection{Problem Formulation}
In order to find a scheduling policy that can improve monitoring performance defined by AoI function, we define the cost function of the $i$-th sensor as $f(\Delta _i^R(k))$. The mapping relationship from the state to the action is defined as $\pi \in \Pi$. Then, the optimization problem over infinite time horizon is formulated as follows:
\begin{problem} \label{prob1}
	\begin{align} \label{10}
		& \min \limits_{\pi \in \Pi} \ \ \limsup \limits_{T \to \infty} \frac{1}{T} \sum \limits_{k=1}^T \sum \limits_{i=1}^N \mathbb{E} [f(\Delta _i^R(k))], \\
		& \ \  \text{s.t. \ \ (3) holds.}
	\end{align}
\end{problem}

\section{The Proposed Scheduling Policy}
In this section, we first formulate the MDP based on the optimization Problem \ref{prob1} in Section III. Then, we investigate the existence and structure of the optimal scheduling policy and, on this basis, develop a structure-informed scheduling policy with low complexity. Finally, the condition for ensuring the stability of monitoring performance is analyzed.
\subsection{MDP Formulation}
We assume that state evolution of AoLI and AoRI are both observable to the MC before each scheduling decision. According to Problem \ref{prob1}, we consider the system state consisting of AoLI $\Delta^L_i(k)$, AoRI $\Delta^R_i(k)$ and channel state $\vartheta(k)$, which varies only relying on the current state and decision instead of historical states. Therefore, the scheduling problem can be regarded as a MDP characterized by a quadruple $(\mathcal{S}, \mathcal{A}, \Pr(\cdot|\cdot, \cdot), c(\cdot, \cdot))$ and details are as follows:

\begin{itemize}
	\item The state space $\mathcal{S}$ includes all possible values of AoLI $\Delta^L_i(k)$, AoRI $\Delta^R_i(k)$ for all sensors and channel state $\vartheta(k)$, denoted as $s(k) = (s^\Delta(k), \vartheta(k))$, where $s^\Delta(k) = (s_1^\Delta(k), s_2^\Delta(k), \dots, s_N^\Delta(k))$ and $s_i^\Delta(k) = (\Delta^L_i(k), \Delta^R_i(k)), i \in \mathcal{N}$.
\end{itemize}

\begin{itemize}
	\item The action space $\mathcal{A}$ consists of all feasible scheduling decisions, which is expressed as $d(k) = ({d_1}(k),{d_2}(k), \ldots ,{d_N}(k)) \in {\left\{ {0,1} \right\}^N}$, $d(k) \in\mathcal{A} $.
\end{itemize}

\begin{itemize}
	\item The transition function is related to the probabilities of random arrivals and successful transmissions. According to the action taken in the current time slot, the state transits to a new state in the next time slot with the following transition probability:
	\begin{align} \label{eq11}
		& \Pr \{ {s(k)\left| {s(k-1),d(k)} \right.} \} \nonumber \\
		& = \Pr \{ \vartheta(k) | \vartheta(k-1) \} \nonumber \\ 
		& \ \ \ \times \Pr \{ {s^\Delta(k) \left| {s^\Delta(k-1), \vartheta(k-1), d(k)} \right.} \} \nonumber \\
		& = \Pr \{ \vartheta(k) | \vartheta(k-1) \}  \nonumber \\
		& \ \ \ \times \prod \limits_{i=1}^N \Pr \{ {s_i^\Delta(k) \left| {s_i^\Delta(k-1), \vartheta(k-1), d_i(k)} \right.} \},
	\end{align}
	where the first term only follows the Markov chain defined in \eqref{eq1}, and it is denoted as 
	\begin{align} \label{eq12}
		& \Pr \{ \vartheta(k) | \vartheta(k-1) \} \nonumber \\
		& = \begin{cases}
			\kappa_{00},     & {\text{if}} \ \ \vartheta(k) = 0, \vartheta(k-1) = 0, \\
			1 - \kappa_{00}, & {\text{if}} \ \ \vartheta(k) = 0, \vartheta(k-1) = 1, \\
			\kappa_{11},     & {\text{if}} \ \ \vartheta(k) = 1, \vartheta(k-1) = 1, \\
			1 - \kappa_{11}, & {\text{if}} \ \ \vartheta(k) = 1, \vartheta(k-1) = 0, \\
			0,               & {\text{otherwise}}.
		\end{cases}
	\end{align}
	
	Then, the second term in the right-hand side (RHS) of \eqref{eq11} is given by
	\begin{align} \label{16}
		& \Pr \{ {s_i^\Delta(k) \left| {s_i^\Delta(k-1), \vartheta(k-1), d_i(k)} \right.} \} \nonumber \\
		& =\begin{cases}
			\lambda_i p_i^\vartheta (k),            & {\text{if}} \ \ \text{case (1)}, \\
			(1 - \lambda_i) p_i^\vartheta (k),      & {\text{if}} \ \ \text{case (2)}, \\
			\lambda_i (1 - p_i^\vartheta(k)),       & {\text{if}} \ \ \text{case (3)}, \\
			(1 - \lambda_i) (1 - p_i^\vartheta(k)), & {\text{if}} \ \ \text{case (4)}, \\
			\lambda_i,                              & {\text{if}} \ \ \text{case (5)}, \\
			1 - \lambda_i,                          & {\text{if}} \ \ \text{case (6)}, \\
			0, & {\text{otherwise}},
		\end{cases}
	\end{align}
	where $p_i^\vartheta (k)$ represents the successful transmission probability under channel state $\vartheta (k-1)$ for sensor $i$. The above six possible cases are listed as follows:
	
	case (1): ${d_i}(k) = 1, s_i^\Delta(k) = (0, {\Delta _i^L}(k-1) + 1)$,
	
	case (2): ${d_i}(k) = 1, s_i^\Delta(k) = (\Delta _i^L (k-1) + 1, {\Delta _i^L}(k-1) + 1)$, 
	
	case (3): ${d_i}(k) = 1, s_i^\Delta(k) = (0, {\Delta _i^R}(k-1) + 1)$,
	
	case (4): ${d_i}(k) = 1, s_i^\Delta(k) = (\Delta _i^L (k-1) + 1, {\Delta _i^R}(k-1) + 1)$,
	
	case (5): ${d_i}(k) = 0, s_i^\Delta(k) = (0, {\Delta _i^R}(k-1) + 1)$,
	
	case (6): ${d_i}(k) = 0, s_i^\Delta(k) = (\Delta _i^L (k-1) + 1, {\Delta _i^R}(k-1) + 1)$.
\end{itemize}

\begin{itemize}
	\item The one-step cost function is obtained form the given action and state. Thus, according to Problem \ref{prob1}, the one-step cost of MDP is defined as $c(s(k),d(k))= \sum\nolimits_{i = 1}^N f(\Delta _i^R(k))$. 
\end{itemize}

Given a scheduling policy $\pi$, the time-average cost function with an initial state $s_0$ is written as
\begin{equation} \label{23}
	\mathcal{C}(s_0,\pi) = \mathop { \limsup }\limits_{T \to \infty } \frac{1}{T} \sum\limits_{k=1}^ T \mathbb{E} \left[ {c(s(k),d(k)) | s_0, \pi} \right]. 
\end{equation}
Hence, the optimization problem is equivalent to find a policy minimizing the cost function \eqref{23}, that is 
\begin{problem} \label{prob2}
	\begin{align} \label{24}
		\min \limits_{\pi  \in \Pi } \ \ \mathcal{C}(s_0,\pi), \ \ \text{s.t. \ (3) holds.}
	\end{align}
\end{problem}

\subsection{Existence and Property of Deterministic Stationary Optimal Policy}
According to the theory present in \cite{guo2006average}, the solution of Problem \ref{prob2} corresponds to the optimal policy $\pi^*$ available within the set of deterministic stationary policies, provided that the time-average cost \eqref{23} is bounded. Assume that $\phi^*$ is the optimal value of $\mathcal{C}(s_0,\pi)$ under any initial value $s_0 \in \mathcal{S}$, then the Bellman equation can be formulated as follows:
\begin{equation} \label{eq26}
	\phi^* + Q(s) = \mathop {\min }\limits_{d \in \mathcal{A}} \left[ {c(s,d) + \sum\limits_{s' \in \mathcal{S}} {\Pr \{s'| {s,d} \}} Q(s')} \right], 
\end{equation}
where $Q(s)$ denotes the action-value function that reflects the expected objective value under the optimal policy $\pi^*$, and $\Pr (s'\left| {s,d} \right.)$ is defined in \eqref{eq11}. To simplify the expression, $s'$ denotes the state of next step, $s$ and $d$ denote $s(k)$ and $d(k)$, respectively. The optimal policy $\pi^*$ is deservedly written as
\begin{equation} \label{eq27}
	\pi^*(s) = \argmin_{d \in \mathcal{A}} \left[ {c(s,d) + \sum\limits_{s' \in \mathcal{S}} \Pr \{s'| {s,d} \} Q(s')} \right].
\end{equation}

In turn, we derive the existence theorem of deterministic stationary optimal policy based on the formulated MDP framework. 
\begin{theorem}\label{the1}
	If Problem \ref{prob2} is feasible, i.e., time-average cost \eqref{23} is bounded, there exist a constant $\phi^*$ and a function $Q(\cdot)$ such that the deterministic stationary optimal policy $\pi^*$ can be obtained by solving Bellman equation \eqref{eq26}.
\end{theorem}
\begin{proof}
	See Appendix \ref{Proof Theorem 1}.
\end{proof}

\begin{remark}
	By reducing the 0-1 knapsack problem, Problem \ref{prob2} can be proven to be NP-hard. When the infinite-horizon cost $\mathcal{C}(s_0,\pi)$ is focused on finite-horizon, there exists a number $\mathcal{\bar C}$ that $\mathcal{C}(s_0,\pi) \leq \mathcal{\bar C}$. The scheduling time for each sensor can be used as a weight. Minimizing the sum of AoI function via scheduling sensors is viewed as selecting knapsacks to maximize object value subject to the total time. Thus, Problem \ref{prob2} is verified as a NP-hard problem. Then, MDP framework is built for this problem and Theorem \ref{the1} establishes the existence of a deterministic stationary optimal policy, whose structure is inherently determined by the action–value function $Q(s)$. Nevertheless, deriving this optimal policy requires solving the recursive relation in \eqref{eq26}, making a closed-form characterization intractable. This motivates the need to investigate the structural property to guide the design of scheduling policy.
\end{remark}

Then, based on the established Bellman equation \eqref{eq26}, the monotonicity analysis of the optimal policy $\pi^*$ will be analyzed. In order to express the iteration process, the action-based cost function of the $n$-th iteration is defined as 
\begin{equation}
	\Theta_n(s, d) = c (s,d) + \sum\limits_{s' \in \mathcal{S}} {\Pr \{s'| {s,d} \}} Q_n(s'). \label{35}
\end{equation}
Define $\tilde \Delta^L = (\Delta_i^L)_{i \in \mathcal{N}}$ and $\tilde \Delta^{R} = (\Delta_i^R)_{i \in \mathcal{N}}$. Then, we can further obtain the following lemma:

\begin{lemma} \label{lemma2}
	Given a channel state $\vartheta$, for two states $s^a = (\tilde\Delta^{L,a}, \tilde \Delta^{R,a}, \vartheta)$ and $s^b = (\tilde\Delta^{L,b}, \tilde\Delta^{R,b}, \vartheta)$, there exists $Q(s^b) \geq Q(s^a)$ if $\tilde\Delta^{L,b}  \succeq  \tilde\Delta^{L,a}$ and $ \tilde \Delta^{R,b} \succeq \tilde \Delta^{R,a}$.
\end{lemma}

\begin{proof}
	See Appendix \ref{Proof lemma 1}.
\end{proof}

Lemma \ref{lemma2} reveals that the action-value function $Q(s)$ increases monotonically with the staleness of information under a fixed channel state. The monotonicity of $Q(s)$ theoretically explains that the optimal policy tends to schedule sensors with more outdated data. Some existing literature introduces a threshold structure to design the optimization algorithm, such as \cite{wei_double_2023}. However, seeking a scheduling policy for multiple sensors possesses a high computational complexity, which motivates us to restrict the search space accordingly. 

\subsection{Structure-Informed Scheduling Policy}
Inspired by the approximate dynamic programming \cite{bertsekas_2007}, we formulate a distributed framework to decompose the optimization problem based on the structural property in Lemma \ref{lemma2} and decrease the complexity via introducing the randomized policy that easy to implement. We first define a scheduling policy $\pi_R$ that sensors are chosen randomly with probability $p_i^R \in (0,1)$, $i \in \mathcal{N}$, where $p_i^R$ is constrained by $\sum_{i = 1}^N p_i^R \leq M$. Then, we can further derive a linear approximated structure of the optimal policy.

\begin{theorem} \label{thm2}
	There exists an approximation of the optimal policy defined in \eqref{eq27}, denoted as follows:
	\begin{align}\label{41}
		\tilde \pi^* (s) \! = \! \argmin_{d \in \mathcal{A}} \! \left[ {c(s,d) \!+\! \sum\limits_{s' \in \mathcal{S}} \! \Pr \{s'| {s,d} \} \! \sum \limits_{i \in \mathcal{N}} Q_i^R(s_i')} \right]
	\end{align}
	where $Q_i^R(s_i)$ is the value function corresponding to sensor $i$ under the randomized scheduling action. 
\end{theorem}

\begin{proof}
	See Appendix \ref{Proof Theorem 2}.
\end{proof}

\begin{remark}
	Through the linear decomposability of action-value function $Q(s)$, Theorem \ref{thm2} approximates the original optimal policy while ensuring better performance than any randomized scheduling policies. Besides, similar to algorithm in \cite{ying_cui_delay-aware_2012}, this approximation can effectively simplify the computation. Given $\bar \Delta_i^L$ and $\bar \Delta_i^R$ as the upper limits of AoLI and AoRI, respectively. It is evident that when $Q(s)$ is transformed into $\sum \nolimits_{i \in \mathcal{N}}Q_i^R(s)$, the computation complexity is reduced to $O(\sum \nolimits_{i \in \mathcal{N}} 2 \bar \Delta_i^R (\bar \Delta_i^L + 1))$ from $O(\prod \nolimits_{i \in \mathcal{N}} 2 \bar \Delta_i^R (\bar \Delta_i^L + 1))$, which is linearly related to the number of sensors rather than exponentially related.
\end{remark}

\begin{algorithm}[t]
	\caption{Structure-Informed Scheduling Policy}
	\label{alg1}
	\begin{algorithmic}[1]
		\renewcommand{\algorithmicrequire}{\textbf{Input:}}
		\renewcommand{\algorithmicensure}{\textbf{Output:}}
		\REQUIRE $N, M, \lambda_i, \bar\Delta_i^L, \bar\Delta_i^R, \varrho, p_i^R, p_i^\vartheta, \kappa_{00}, \kappa_{11}$ and $n_{max}$.
		\ENSURE Scheduling policy $\tilde \pi^*$.
	    \FOR{$i = 1$ to $N$}
		\STATE Initialize $Q_{i,n}^R(s_i) \leftarrow 0, \forall s_i \in \mathcal{S}_i, n\leftarrow1$;
		\REPEAT
		\FOR{all $s_i \in \mathcal{S}_i$}
		\STATE $\begin{aligned} & Q_{i,n+1}^R(s_i) \leftarrow  f(\Delta_i^R) \\ & + {\sum\nolimits_{s_i' \in \mathcal{S}_i} {\mathbb{E}^{\pi_R}[\Pr \{s_i'\left| {s_i,d_i} \right.}\}] Q_{i,n+1}^R(s_i')}; \end{aligned}$
		\STATE $Q_{i,n+1}^R(s_i) \leftarrow Q_{i,n+1}^R(s_i) - Q_{i,n+1}^R(s_i^r)$;
		\ENDFOR
		\STATE $n \leftarrow n+1$;
		\UNTIL{$\Vert Q_{i,n}^R(s_i) - Q_{i,n-1}^R(s_i^r) \Vert  \leq \varrho$ or $n = n_{max}$;}
		\ENDFOR
		\FOR{all $s \in \mathcal{S}$}
		\IF{$\exists i \in \mathcal{N}$ and $s' \in \mathcal{S}$ such that $\tilde \pi^*(s') = d \in \mathcal{A}$ with $d_i = 1, \tilde\Delta^L = {{}\tilde\Delta^L}', \vartheta = \vartheta'$ and $\Delta_i^R \geq {\Delta_i^R}'$ only for the $i$-th sensor,}
		\STATE $\tilde \pi^*(s') \leftarrow d$;
		\ELSE 
		\STATE Obtain $\tilde \pi^*(s')$ by \eqref{41};
		\ENDIF
		\ENDFOR
		\RETURN Scheduling policy $\tilde \pi^*$.
		
	\end{algorithmic}
\end{algorithm}

Some works have given the threshold structure for the optimal policy hoping to reduce the computational cost \cite{wu_optimal_2020}. Next, we further analyze the similar structural property of the approximated policy proposed in Theorem \ref{thm2}. 

\begin{theorem} \label{the3}
	The proposed approximated optimal policy possesses a threshold structure with respect to (w.r.t.) $\Delta_i^R$, such that there exists the following result for the $i$-th sensor, $i \in \mathcal{N}$:
	\begin{equation} \label{44}
		\tilde \pi_i^* =
		\begin{cases}
			1, & \text{if} \ \ \Delta_i^R \geq {\Delta_{0,i}^R}^*, \vartheta = 0,\\
			1, & \text{if} \ \ \Delta_i^R \geq {\Delta_{1,i}^R}^*, \vartheta = 1,\\
			0, & \text{otherwise},
		\end{cases} 
	\end{equation}
	where ${\Delta_{0,i}^R}^*$ and ${\Delta_{1,i}^R}^*$ are denoted as
	\begin{align*}
		& {\Delta_{0,i}^R}^* \! \triangleq \! \min \! \left\{ {{\Delta _i^R}\left| {\Theta^\dag (s,d) \! \leq \! \Theta^\dag (s,d'), \vartheta = 0, d \ne d', d_i = 1} \right.} \right\} \!, \\
		& {\Delta_{1,i}^R}^* \! \triangleq \! \min \! \left\{ {{\Delta _i^R}\left| {\Theta^\dag (s,d) \! \leq \! \Theta^\dag (s,d'), \vartheta = 1, d \ne d', d_i = 1} \right.} \right\} \!, \\
		& \Theta^\dag (s,d) \! \triangleq \! c(s,d) + \sum\nolimits_{s' \in \mathcal{S}} \Pr (s'\left| {s,d} \right.) \sum \nolimits_{i \in N} Q_i^R(s_i').
	\end{align*}
\end{theorem}
\begin{proof}
	See Appendix \ref{Proof Theorem 3}.
\end{proof}

Theorem \ref{the3} indicates that a sensor is to be scheduled for transmission once AoRI exceeds a threshold ${\Delta_i^R}^*$, which reflects the urgency to refresh the outdated information. Moreover, this action demonstrates a certain persistence that if $\tilde \pi^* (s) = 1$ under current $\Delta_i^R$ , than for any $(\Delta_i^R)' > \Delta_i^R$ there exists $\tilde \pi^* (s') = 1$. This provides an intuitive scheduling rule that allows efficient implementations without solving the full MDP. The proposed structure-informed scheduling policy (SISP) is shown in Algorithm \ref{alg1}, which operates in two main phases. In lines 1-10, the relative value iteration algorithm is employed to update the $i$-th decomposed value function $Q_i^R(s_i)$ proposed in Theorem \ref{thm2}, where $Q_{i,n}^R(s_i)$ is the $n$-th iteration and a reference state $s_i^r$ is used for normalization in line 6 to ensure numerical stability. Lines 11-17 constructs the optimal policy $\tilde \pi^*$, where the threshold structure proved in Theorem \ref{the3} is utilized. It is obvious that the scheduling decision is made directly if the condition in line 12 is meet and this structure-informed pruning significantly reduces the computational overhead.

\subsection{Stability Analysis of the Monitoring Performance}
In this section, we analyze the conditions required to ensure the convergence of monitoring performance. Due to the uncertainties introduced by random data arrivals and fluctuating channel states, it is difficult to maintain consistently successful update transmissions, which may worsen the monitoring performance. Therefore, it is necessary to identify the range of transmission probabilities and random arrival rates that can guarantee stable monitoring, thereby providing a guidance for the design of real remote monitoring networks.

For the scenario of a single sensor system, studies in \cite{wang2021remote} and \cite{wei_double_2023} have analyzed the conditions for ensuring the stability of remote state estimation under Markov channel states, but this class of theoretical results cannot be directly applied to the network system consisting of multiple sensors. In \cite{chang_lightweight_2024}, the stability conditions of remote state estimation in multi-sensor systems have been investigated, while channel is limited to be static and the impact of random data arrivals is neglected. Referring to the threshold structure proposed in Theorem \ref{the3}, we utilize the limitation distribution of the value of $\Delta_i^R$ to derive the sufficient and necessary condition for the stable monitoring performance. To pursue a more universal result, we consider the more complex AoI function introduced in Example \ref{exam1}.

\begin{figure}[t]
	\renewcommand{\figurename}{Fig.}
	\centering
	\subfloat[Sensor 1 with $\lambda_1 = 0.9$.\label{fig3.1}]{\includegraphics[width=0.35\textwidth]{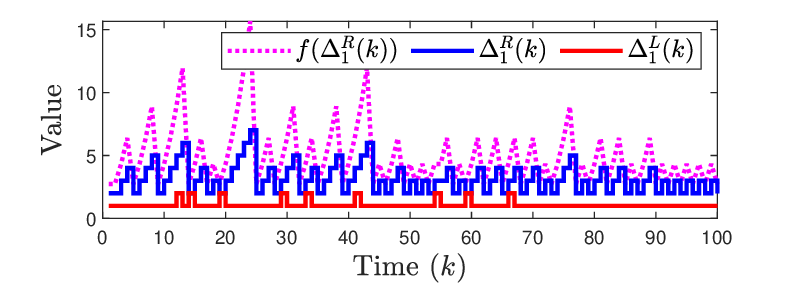}} 
	
	\subfloat[Sensor 2 with $\lambda_2 = 0.9$.\label{fig3.2}]{\includegraphics[width=0.35\textwidth]{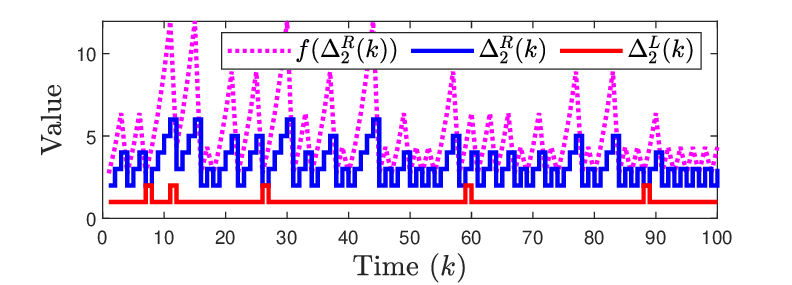}}
	
	\subfloat[Sensor 2 with $\lambda_2 = 0.5$.\label{fig3.3}]{\includegraphics[width=0.35\textwidth]{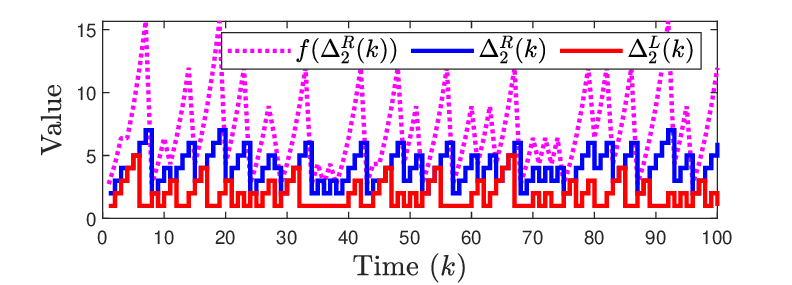}}
	\caption{The evolution of AoI function, AoLI and AoRI for two sensors.}
	\label{fig3}
\end{figure}

\begin{theorem} \label{thm4}
	For the AoI function defined in \eqref{3.8}, the time-average cost $\mathcal{C}(s_0,\pi)$ is bounded under the proposed scheduling policy $\hat \pi^*$ if and only if the following condition is satisfied:
	\begin{align} \label{eq21}
		\rho (\Omega (I - \lambda_i \mathcal{P}_i)) < \frac{1}{\rho^2 (A_i)}, \forall i \in \mathcal{N},
	\end{align}
	where $\Omega = \begin{bmatrix}
		\kappa_{00} & \kappa_{01}\\
		\kappa_{10} & \kappa_{11} \\
	\end{bmatrix}$ and $\mathcal{P}_i = \text{diag} \{p_i^0, p_i^1\}$.
\end{theorem}
\begin{proof}
	See Appendix \ref{Proof Theorem 4}.
\end{proof}

\begin{remark}
	According to the form of \eqref{eq21}, we observe that monitoring performance can only be maintained if transmission channels are sufficiently reliable to match the rate at which data arrive, providing the guidance for the design of wireless monitoring networks. Using the same method, a corollary is further derived for the AoI function defined in Example \ref{exam2}. Moreover, this result can be extended to scenarios where data arrivals follow a Markov process.
\end{remark}

\begin{corollary} \label{cor1}
	For the AoI function defined in \eqref{3.9}, the time-average cost $\mathcal{C}(s_0,\pi)$ is bounded under the proposed scheduling policy $\hat \pi^*$ if and only if the following condition is satisfied:
	\begin{equation} \label{4.26}
		\rho (\Omega (I - \lambda_i \mathcal{P}_i)) < \frac{1}{e^r}, \forall i \in \mathcal{N}.
	\end{equation}
\end{corollary}

In the following, we consider a more complex case where the data arrival process evolves under a Markov chain. Assume that the probability of no data arrival at the current time slot is $\tilde \lambda_i$ when no data arrived in the previous time slot, i.e., $\Pr \{\gamma_i^L (k) = 0 | \gamma_i^L (k - 1) = 0\} = \tilde \lambda_i$. When data arrived in the previous time slot, the probability of data arrival at the current time slot is set as $\bar \lambda_i$, i.e., $\Pr \{\gamma_i^L (k) = 1 | \gamma_i^L (k - 1) = 1\} = \bar \lambda_i$. We can get a new corollary of Theorem \ref{thm4}. 

\begin{corollary}
	For the AoI function defined in \eqref{3.8}, when the data arrival is a Markov process with transition matrix $\begin{bmatrix} \tilde \lambda_i & 1 - \tilde \lambda_i\\ 1 - \bar \lambda_i & \bar \lambda_i \\ \end{bmatrix}$, the time-average cost $\mathcal{C}(s_0,\pi)$ is bounded under the proposed scheduling policy $\hat \pi^*$ if and only if the following condition is satisfied:
	\begin{align}
		\rho (\Omega (I -  \hat \lambda_i \mathcal{P}_i)) < \frac{1}{\rho^2 (A_i)}, \forall i \in \mathcal{N},
	\end{align}
	where $\hat \lambda_i = \min \{1 - \tilde \lambda_i, \bar \lambda_i\}$.
\end{corollary}
\begin{proof}
	See Appendix \ref{Proof Corollary 2}.
\end{proof}

\begin{figure*}[!t]
	\renewcommand{\figurename}{Fig.}
	\centering
	
	\subfloat[$\lambda_1 = 0.9$, $\lambda_2 = 0.5$, $\vartheta = 0$]{%
		\includegraphics[width=0.3\textwidth]{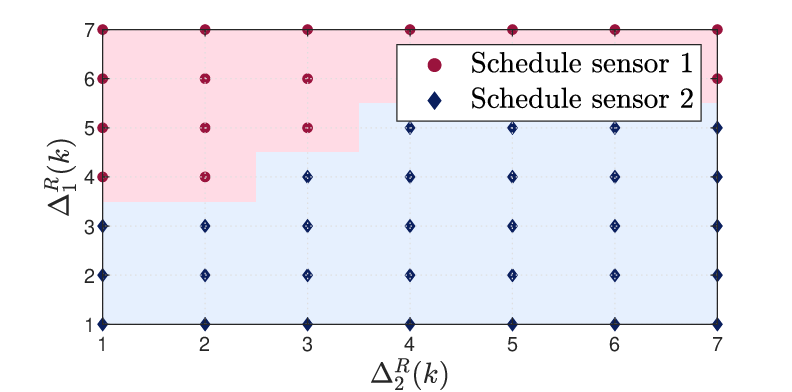}%
		\label{fig4.1}%
	}
	\subfloat[$\lambda_1 = 0.9$, $\lambda_2 = 0.5$, $\vartheta = 1$]{%
		\includegraphics[width=0.3\textwidth]{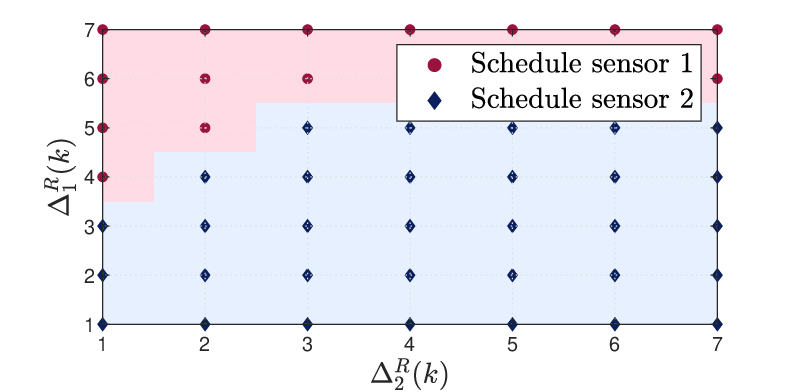}%
		\label{fig4.2}%
	}
	\subfloat[$\lambda_1 = \lambda_2 = 0.9$, $\vartheta = 0$]{%
		\includegraphics[width=0.3\textwidth]{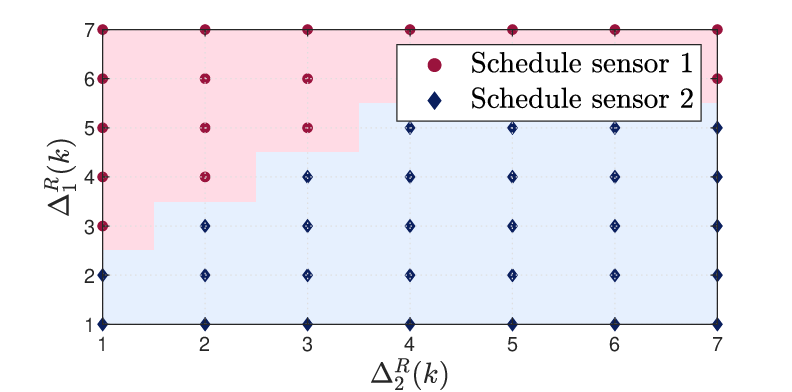}%
		\label{fig4.3}%
	}
	
	\subfloat[$\lambda_1 = 0.9$, $\lambda_2 = 0.5$]{%
		\includegraphics[width=0.3\textwidth]{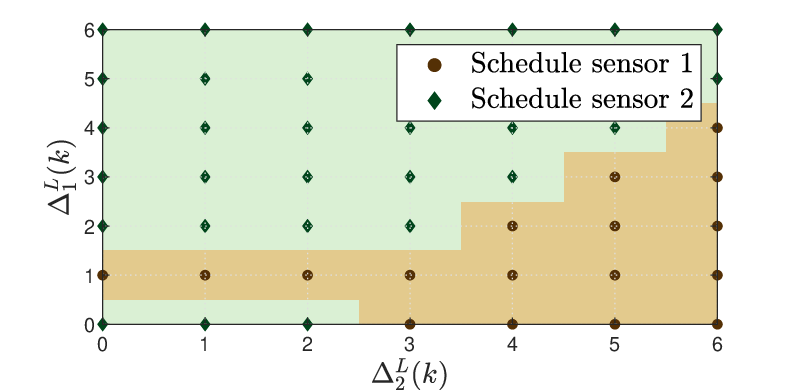}%
		\label{fig4.4}%
	}
	\subfloat[$\lambda_i = 0.9$]{%
		\includegraphics[width=0.3\textwidth]{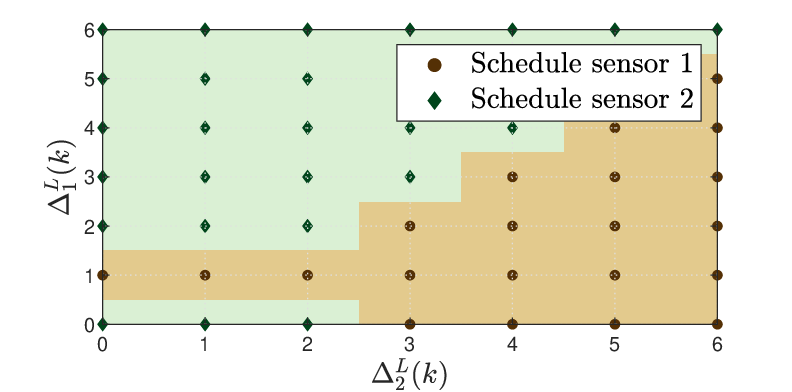}%
		\label{fig4.5}%
	}
	\caption{Scheduling structure of SISP in the two sensor system under different cases.}
	\label{fig4}
\end{figure*}

\begin{figure}[t]
	\renewcommand{\figurename}{Fig.}
	\centering
	\subfloat[$\lambda_1 = \lambda_2 = 0.9$.\label{fig5.1}]{\includegraphics[width=0.35\textwidth]{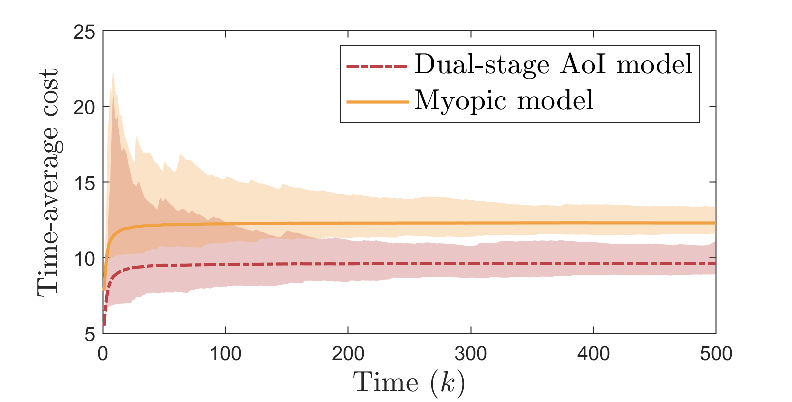}} 
	
	\subfloat[$\lambda_1 = 0.9, \lambda_2 = 0.5$.\label{fig5.2}]{\includegraphics[width=0.35\textwidth]{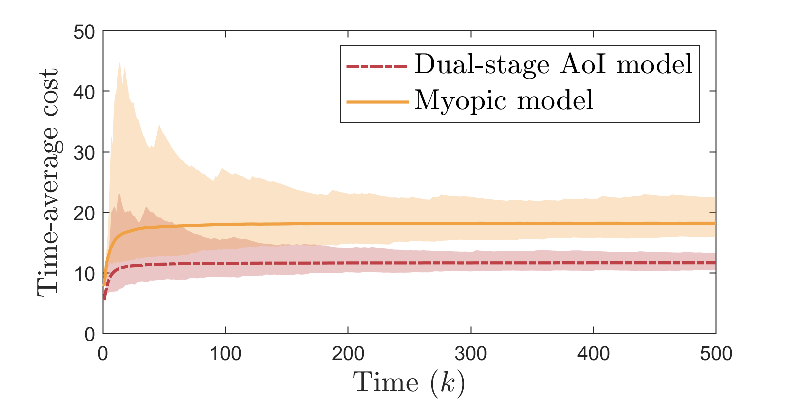}}
	\caption{Comparison of the time-average cost for myopic model and dual-stage AoI model.}
	\label{fig5}
\end{figure}

\begin{figure}[t]
	\renewcommand{\figurename}{Fig.}
	\centering
	\subfloat[$\lambda_1 = \lambda_2 = 0.9$.\label{fig6.1}]{\includegraphics[width=0.35\textwidth]{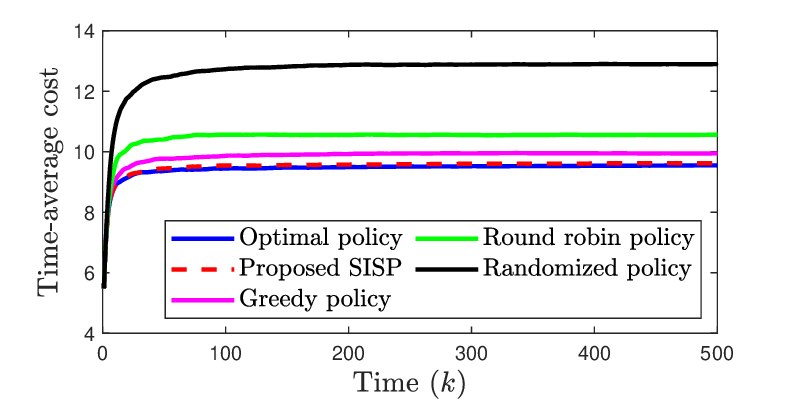}} 
	
	\subfloat[$\lambda_1 = 0.9, \lambda_2 = 0.5$.\label{fig6.2}]{\includegraphics[width=0.35\textwidth]{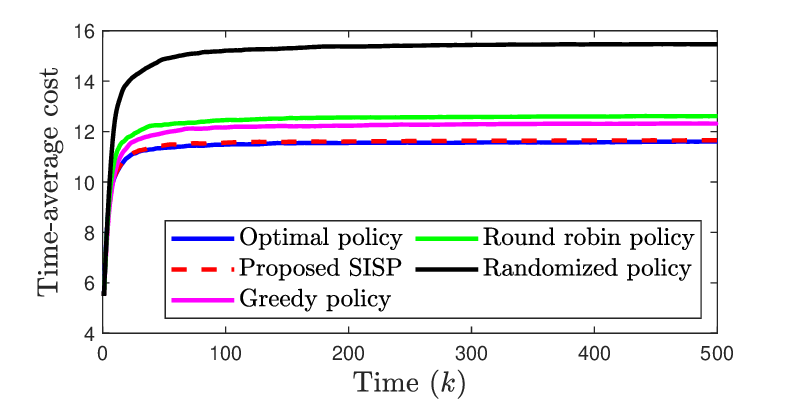}}
	\caption{Performance comparison of different scheduling policies for AoI function in Example 2.}
	\label{fig6}
\end{figure}

\begin{figure}
	\renewcommand{\figurename}{Fig.}
	\centering
	\subfloat[$\lambda_1 = \lambda_2 =  0.8$.\label{fig7.1}]{\includegraphics[width=0.35\textwidth]{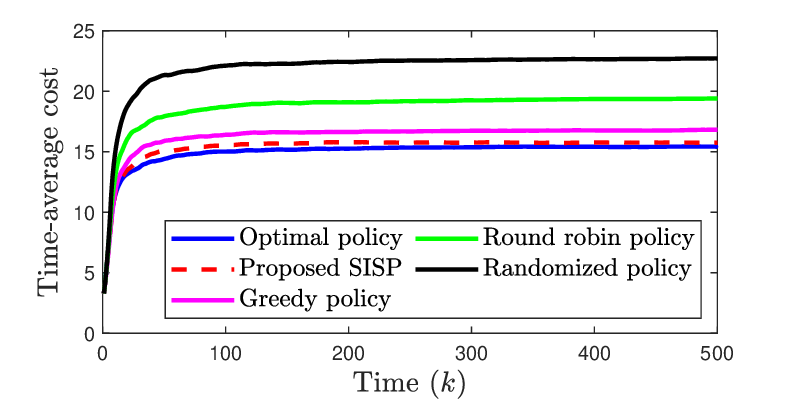}} 
	
	\subfloat[$\lambda_1 = \lambda_2 =  0.6$.\label{fig7.2}]{\includegraphics[width=0.35\textwidth]{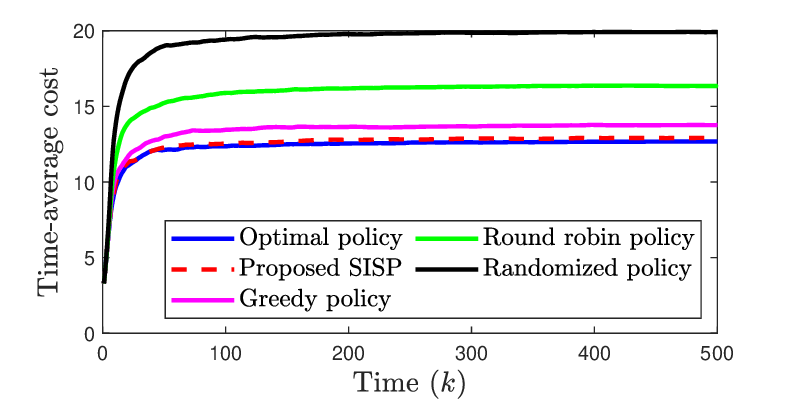}}
	\caption{Performance comparison of different scheduling policies for AoI function in Example 1.}
	\label{fig7}
\end{figure}

\section{Simulation Study}
To demonstrate the performance of the proposed scheduling strategy, we conduct simulation experiments using the AoI functions provided in Examples \ref{exam2} and \ref{exam1}. First, in the two-sensor model, we verify the effectiveness of the two-stage AoI mechanism and analyze the scheduling structure. Then, in the more complex multi-sensor model, we further validate the superiority of Algorithm \ref{alg1} through comparative experiments.

\subsection{Demonstration of State Evolution and Scheduling Structure}
To provide more general results, we adopt the AoI function \eqref{3.8} from Example 2 as the model for the wireless monitoring networks.The system parameters are considered as follows:
\begin{align*}
	A_i = \begin{bmatrix} 1.1 & 0.5\\ 0 & 0.2 \\ \end{bmatrix}, C_i = \begin{bmatrix} 1 & 1 \end{bmatrix}, \Sigma_i^\omega = \begin{bmatrix} 1 & 0\\ 0 & 1 \\ \end{bmatrix},
\end{align*}
and the covariance of measurement noise is $0.8$. Then, the steady-state error covariance $\bar P_i$ is computed as 
\begin{align*}
	\bar P_i = \begin{bmatrix} 0.9038 & -0.5175\\ -0.5175 & 0.7464 \\ \end{bmatrix}.
\end{align*}
For the two-state memory-based channel model, we suppose the successful transmission probability as $p_i^0 = 0.5$ for bad state and $p_i^1 = 1$ for good state, respectively. The transition probabilities as set as $\kappa_{00} = 0.5, \kappa_{01} = 0.5, \kappa_{10} = 0.2, \kappa_{11} = 0.8$, reflecting a fading channel scenario typical in industrial mobility environments. In order to avoid computational complexity growing indefinitely, the state space truncation is given by $\bar \Delta_i^R = 7$.  According with the convergence condition in Theorem \ref{thm4}, we first assume that random data arrival rate $\varepsilon_i = 0.9$ for both sensors. The evolution of AoI function $f(\Delta_i^R(k))$ is plotted in Fig. \ref{fig3}, where the evolution $\Delta_i^R(k)$ and $\Delta_i^L(k)$ are also given. 

In order to reveal the dynamic properties of different arrival rates, we define $\lambda_1 = \lambda_2 = 0.9$ in Fig. \ref{fig3.1} and \ref{fig3.2}, and $\lambda_2 = 0.5$ in Fig. \ref{fig3.3}, which obviously meets the convergence condition presented in Theorem \ref{thm4}. From Fig. \ref{fig3.1} and \ref{fig3.2} we can find that the update of AoRI depends on AoLI and its value is updated only when the scheduling is successful, which corresponds to evolution law in Fig. \ref{fig2}. Moreover, the tendency of $f(\Delta_i^R(k))$ is directly related to AoRI and gradually converges to be stable. We then investigate the case where the arrival rate of monitoring data in sensor 2 is $\lambda_2 = 0.5$. It is easy to observe Fig. \ref{fig3.3} that the update rate of AoLI is significantly reduced, the value of AoRI does not increase notably due to the proposed scheduling policy.

Next, we continue to use the previously defined parameters to investigate the scheduling structure. In Fig. \ref{fig4}, we give scheduling results of SISP under different settings of $\lambda_i$. As shown in Fig. \ref{fig4.1} and \ref{fig4.2}, the scheduling results exhibit a threshold structure w.r.t. AoRI under both $\vartheta = 0$ and $\vartheta = 1$, verifying the proposed Theorem \ref{the3}, where AoLI is fixed as $\Delta_1^L(k) = 7, \Delta_2^L(k) = 6$. For sensor $i \in \{1,2\}$, the larger the value of $\Delta_i^R(k)$, the more likely the sensor $i$ is to be scheduled. Besides, by comparing Fig. \ref{fig4.1} and Fig. \ref{fig4.2}, the frequency of scheduling sensor 2 is higher under good channel state because the lower $\lambda_2$ makes updates to $\Delta_i^R(k)$ more urgent for maintaining freshness. Similarly, when $\lambda_2$ is increased to $\lambda_1 = 0.9$, as shown in Fig. \ref{fig4.3}, the frequency of scheduling sensor 2 decreases. Fig. \ref{fig4.4} and \ref{fig4.5} illustrate the scheduling decisions under different AoLI values, with AoRI fixed at $\Delta_1^R(k) = 2, \Delta_2^R(k) = 3$, which clearly shows that there exists no threshold structure w.r.t. AoLI. The results in Fig. \ref{fig4.4} and \ref{fig4.5} also confirm that sensor 2 is easier to be scheduled under a relatively lower data arrival rate. In order to display the superiority of the proposed dual-stage AoI model under random data arrivals, we compare the time-average cost with the traditional single-stage AoI model, which neglects the AoLI and is thus called myopic model. Fig. \ref{fig5} gives the contrast results, it is evident that dual-stage AoI model outperforms myopic model in terms of monitoring performance and stability.  

\begin{figure*}[!t]
	\renewcommand{\figurename}{Fig.}
	\centering
	\subfloat[$\lambda_1 = \lambda_2 = \lambda_3$, $M = 1$. \label{fig8.1}]{\includegraphics[width=0.35\textwidth]{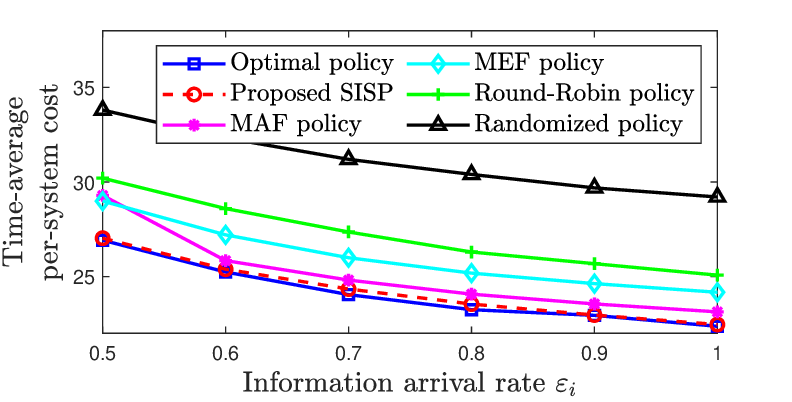}} 
	\subfloat[$\lambda_2 = \lambda_3 = 0.5$, $M = 1$.\label{fig8.2}]{\includegraphics[width=0.35\textwidth]{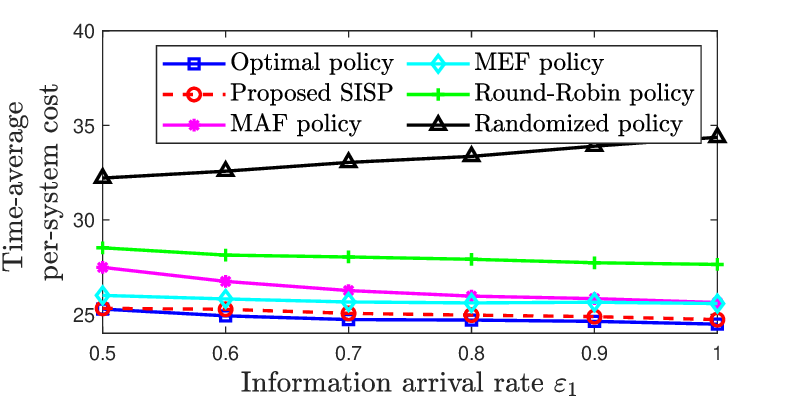}}
	
	\subfloat[$\lambda_1 = \lambda_2 = \lambda_3$, $M = 2$.\label{fig8.3}]{\includegraphics[width=0.35\textwidth]{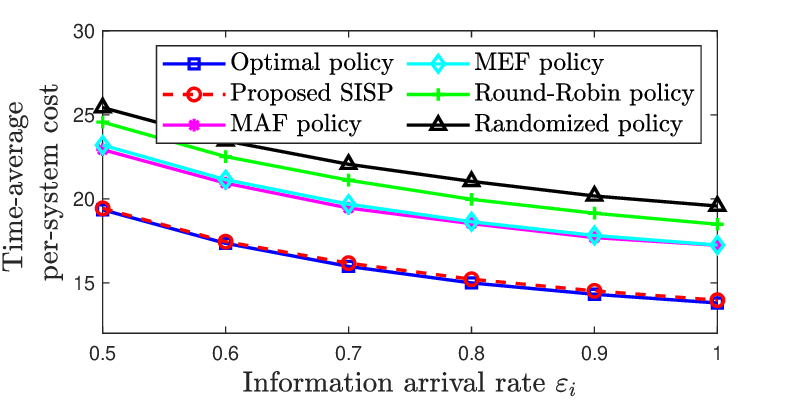}}
	\subfloat[$\lambda_2 = \lambda_3 = 0.5$, $M = 2$.\label{fig8.4}]{\includegraphics[width=0.35\textwidth]{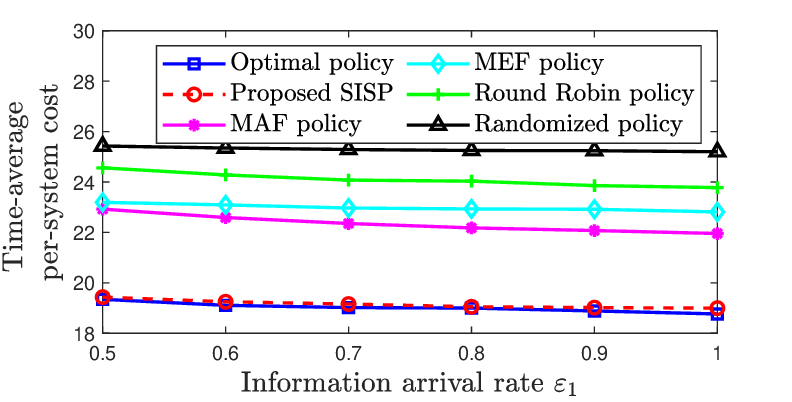}}
	\caption{Performance comparison of different scheduling policies in heterogeneous systems.}
	\label{fig8}
\end{figure*}

\begin{figure*}
	\renewcommand{\figurename}{Fig.}
	\centering
	\subfloat[case 1.\label{fig9.1}]{\includegraphics[width=0.25\textwidth]{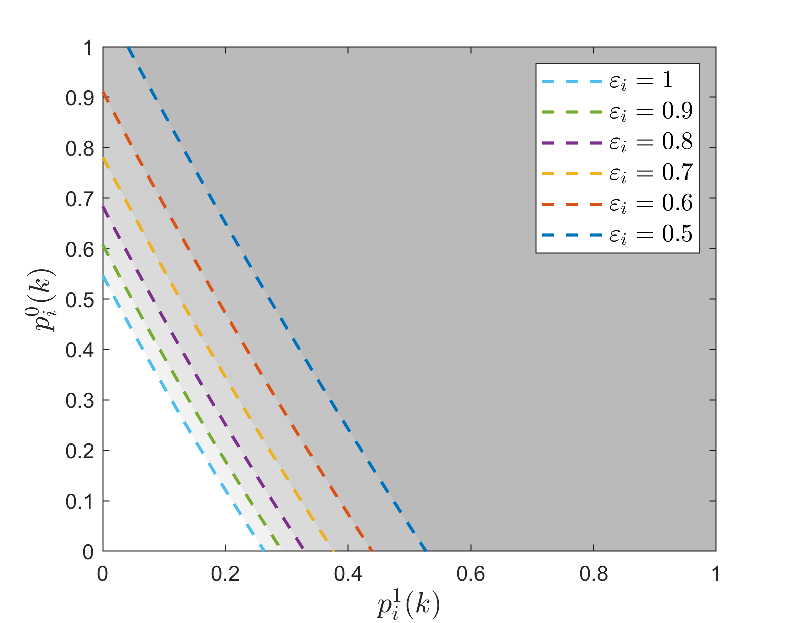}} 
	\subfloat[case 2.\label{fig9.2}]{\includegraphics[width=0.25\textwidth]{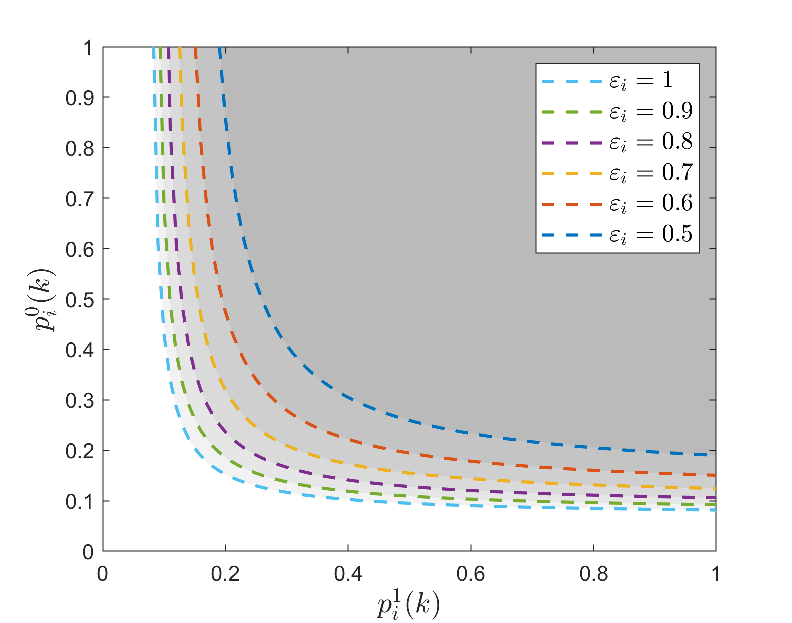}}
	\subfloat[case 3.\label{fig9.3}]{\includegraphics[width=0.25\textwidth]{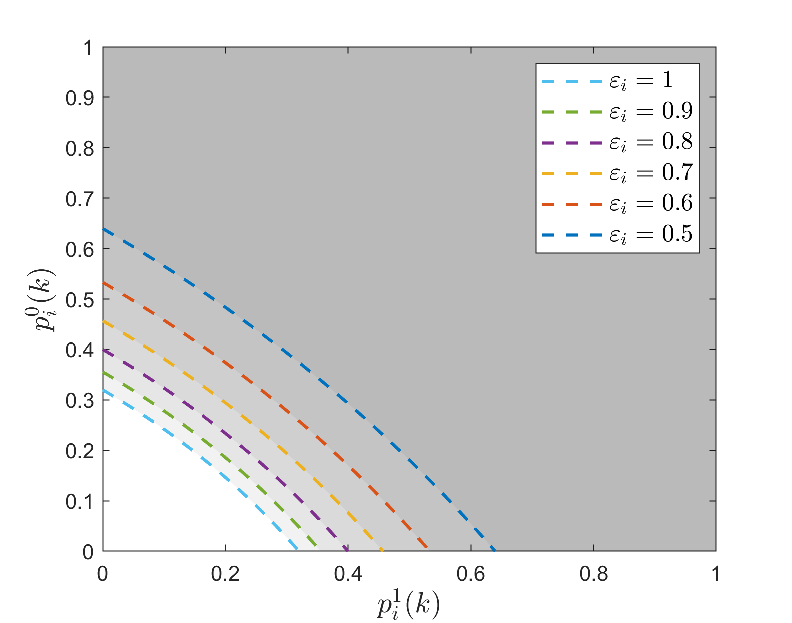}}
	\subfloat[case 4.\label{fig9.4}]{\includegraphics[width=0.25\textwidth]{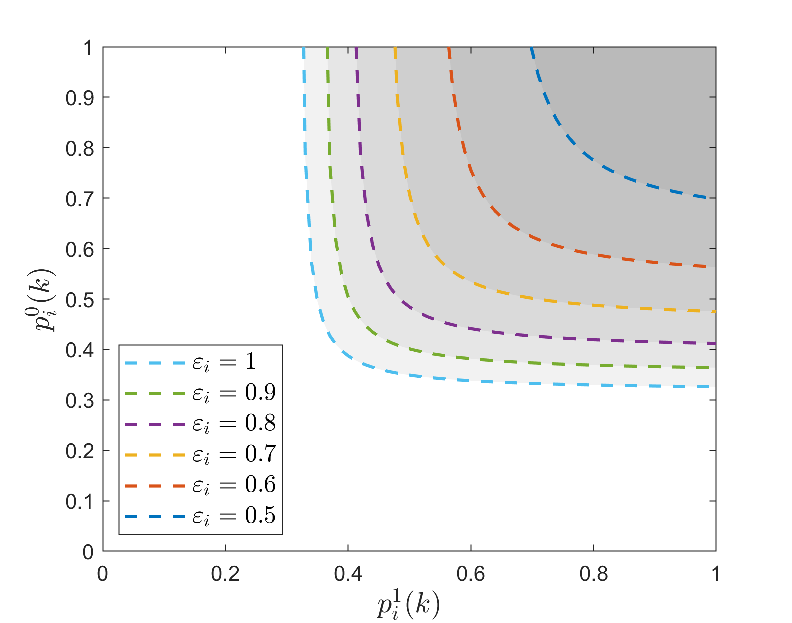}}
	
	\caption{Feasible region of transmission probabilities for two states under different random arrival rates.}
	\label{fig9}
\end{figure*}

\subsection{Performance Comparison with Different Policies}
First, four common scheduling policies are introduced for the comparison. The maximum age first (MAF) policy \cite{zhu_aoi_2021}:  Select the sensor with the largest AoRI to transmit in each time slot. The maximum error first (MEF) policy \cite{han2017optimal}: Select the sensor with the largest error to transmit in each time slot, similar to the greedy policy. The Round-Robin policy \cite{wu_remote_2024}: Schedule sensors in turn each time based on the established sequence. The randomized policy: Schedule sensors randomly in certain probabilities, which are considered to be proportional to the information arrival rates of corresponding sensors, i.e., $p_i^R = \vartheta_i / \sum \nolimits_{j=1}^N \vartheta_j$. For homogeneous systems, the MAF and MEF policies yield identical performance and can be unified into the Greedy policy. In addition, the optimal policy is also considered to demonstrate the near-optimality of SISP. For the AoI function in Example \ref{exam1}, we follow parameter settings in the previous subsection and conduct 500 Monte Carlo simulations, and results are plotted in Fig. \ref{fig6}. It is obvious that SISP outperforms the four other methods significantly with near-optimal performance and the smaller value of $\vartheta_i$ explicitly leads to a greater performance gap. It should also be noted that the higher $\vartheta_i$ results in lower scheduling costs, which is intuitive, as more frequent updates lead to fresher information. Without loss of generality, we consider the AoI function of Example \ref{exam2}, where $r$ is defined as $0.5$ and other settings remain unchanged. Fig. \ref{fig7} shows the similar comparison results with $\varepsilon_i = 0.8$ and $0.6$. 

Then, we further examine scheduling performance in a network composed of three heterogeneous systems, where the first system model is the same as previously considered, while the second and third ones are modeled as two-link planar industrial robots with a sampling interval of $15ms$ \cite{leong_transmission_2019}. The corresponding system parameters are given as follows:
\begin{align*}
	& A_2 = A_3 = \begin{bmatrix} 1.0058 & 0.0150 & -0.0016 & 0\\ 0.7808 & 1.0058 & -0.2105 & -0.0016 \\ -0.1160 & 0.0000 & 1.0077 & 0.0150 \\ -0.7962 & -0.0060 & 1.0294 & 1.0077 \\ \end{bmatrix},\\
	& \Sigma_1^\omega = \Sigma_2^\omega = \sigma \sigma^T, \sigma = \begin{bmatrix} 0.003 & 1 & -0.005 & -2.15 \\ \end{bmatrix}, \\
	& C_1 = C_2 = \begin{bmatrix} 1 & 0 & 0 & 0 \\ 0 & 0 & 1 & 0 \\ \end{bmatrix}.
\end{align*}
Thus, the steady state of error covariance is computed as 
\begin{align*}
	\bar P_2 = \bar P_3 = \begin{bmatrix}
		0.0003 & 0.0077 & -0.0002 & -0.0148 \\
		0.0077 & 1.3150 & -0.0130 & -2.8174 \\
		-0.0002 & -0.0130 & 0.0007 & 0.0289 \\
		-0.0148 & -2.8174 & 0.0289 & 6.0613 \\
	\end{bmatrix}.
\end{align*}
The transmission probabilities of two channel states are supposed as $p_i^0 = 0.4, p_i^1 = 0.9$, and $\kappa_{00} = 0.4, \kappa_{01} = 0.6, \kappa_{10} = 0.3, \kappa_{11} = 0.7$. The influence of various data arrival rates and the number of scheduled sensors on monitoring performance is studied in different policies. All simulations in Fig. \ref{fig8} shows the time-average per-system cost that our proposed SISP get a good lead in the comparison, and it nearly achieves performance identical to that of the optimal policy. By comparing results of $m = 1$ and $m = 2$, the performance gap get smaller when more sensors are scheduled. Moreover, compared with the case where all sensors share the same data arrival rate, the performance varies more smoothly when only one sensor’s data arrival rate changes. Specially, when only one sensor is scheduled, the randomized policy even cause the cost to increase as the $\lambda_1$ increases, because the AoI function of sensor 1 contributes a small portion to the total cost. This further verifies the practicality of our scheduling policy.

To study the influence of data arrival rates to the feasible region of channel states when there exist a memory-based channel, we present several sets of results in Fig. \ref{fig9} illustrating the relationship between $\lambda_i$ and $p_i^\vartheta$. According to Theorem \ref{thm4}, the feasible value of $p_i^\vartheta$ is affected by $\rho(A_i)$. For case 1, we set $\rho(A_i) = 1.1$, and $\kappa_{00} = 0.4, \kappa_{01} = 0.6, \kappa_{10} = 0.3, \kappa_{11} = 0.7$. As shown in Fig. \ref{fig9.1}, the feasible region becomes smaller as the decrease of $\lambda_i$. Then, we adjust transition probabilities of channel states in case 2, that is $\kappa_{00} = 0.9, \kappa_{01} = 0.1, \kappa_{10} = 0.1, \kappa_{11} = 0.9$, and results are provided in Fig. \ref{fig9.2}. In this case, the channel tends to remain in its previous state, i.e., the channel memory becomes stronger, and thus the feasible region of $p_i^\vartheta$ value shrinks significantly. Utilizing Corollary \ref{cor1}, the AoI function \eqref{3.9} is also considered with $r = 0.5$ and it derives a smaller region due to $e^r > \rho^2(A_i)$, which is plotted in Fig. \ref{fig9.4}. On the contrary, the region becomes larger when transition probabilities are changed to case 3 with $\kappa_{00} = 0.2, \kappa_{01} = 0.8, \kappa_{10} = 0.8, \kappa_{11} = 0.2$, given in Fig. \ref{fig9.3}. It means that more frequent state switching will lead to more feasible value of $p_i^\vartheta$.  

\section{Conclusion}
In this paper, we addressed the transmission scheduling problem for optimizing the AoI function modeled based on monitoring performance, where monitoring data arrives randomly and channel state exhibits temporal correlation. We formulated a dual-AoI model to describe different evolution stages, and derived a low-complexity scheduling algorithm based on the structural properties of the MDP. Theoretical analysis revealed that the optimal policy follows a channel state-dependent threshold structure w.r.t. AoRI and established the necessary and sufficient condition for monitoring stability. Results of simulations demonstrate that the proposed policy significantly outperforms existing baselines. Specially, dual-stage AoI model is more effective compared with single-stage model for random data arrivals. Notably, we observe that the threshold structure applies only to AoRI, and sensors with higher arrival rates are prioritized. Furthermore, higher data arrival rates effectively expand the stability region, allowing the system to tolerate lower channel qualities. Future research will focus on extending our conclusions to multi-hop networks and incorporating learning-based methods to design intelligent scheduling methods for control aware communication. 

\appendices
\section{Proof of Theorem 1} \label{Proof Theorem 1}
To confirm the existence of deterministic stationary policy for the Problem \ref{prob2}, we should verify the conditions proposed in \cite[Theorem 2]{cavazos1992comparing}, listed as follows: 
\begin{enumerate}
	\item There exists a stationary policy $\pi$ inducing an irreducible positive recurrent Markov chain such that the average cost $\mathcal{C}(s_0,\pi)$ is bounded.
	\item For any positive number $r$, the set $\{s|c(s,\pi) < r\}$ is finite.
\end{enumerate}

According to the formulated optimization problem, the action space is finite and the one-step cost is bounded below by 0. Besides, if we suppose that there exists a policy $\pi$ with the bounded average cost, such that 
\begin{equation}
	\lim \limits_{T \to \infty} \frac{1}{T} \sum \limits_{k = 1}^T \mathbb{E}[c(s(k),d(k))] < \infty.
\end{equation}
Therefore, the condition 1) holds. The AoI function $f(\Delta_i^R(k))$ considered in this paper is nondecreasing with $\Delta_i^R$, thus the number of $s$ is finite for $c(s,\pi) < r$ when $r > 0$, and the condition 2) is verified. Then, Theorem \ref{the1} is proved.

\section{Proof of Lemma 1} \label{Proof lemma 1}
By iterating the action-value function, $Q (s)$ of the $(n+1)$-th iteration can be described as below for any $s$:
\begin{equation}
	Q_{n+1}(s) = \min \limits_{d \in \mathcal{A}}  \Theta_{n+1}(s,d) - \min \limits_{d \in \mathcal{A}}  \Theta_{n+1}(s_r,d) \label{36}
\end{equation}
where $s_r$ is a definite state. For any initial value of $Q_1(s)$, we have $\lim \limits_{n \to \infty} Q_n(s) = Q(s)$. The optimal policy of $n$th iteration is further denoted as $\pi_n^*(s) = \arg \min \nolimits_{d \in \mathcal{A}} \Theta_{n+1}(s,d)$ from \eqref{eq27}, where $\pi_n^*(s) = (\pi_{i,n}^*(s))_{i\in \mathcal{N}}$. Then, we use the  mathematical induction to verify Lemma \ref{lemma2}.

Given the initialized $Q_1(s)$, assume $Q_1(s^b) \geq V_1(s^a)$ holds, as well as $Q_n(s^b) \geq Q_n(s^a)$. For $(n+1)$-th iteration, utilizing the form in \eqref{36} can derive that for $s^a$ there exists
\begin{equation}
	Q_{n+1}(s^a) = \Theta_{n+1}(s^a,\pi_n^*(s^a)) - \Theta_{n+1}(s_r,\pi_n^*(s_r)),
\end{equation}
which can infer the following inequality via the policy optimality
\begin{align}
	Q_{n+1}(s^a) & \leq \Theta_{n+1}(s^a,\pi_n^*(s^b)) - \Theta_{n+1}(s_r,\pi_n^*(s_r)) \nonumber \\
	& = c (s^a,\pi_n^*(s^b)) + \sum\limits_{{s^a}' \in \mathcal{S}} {\Pr \{{s^a}'| {s^a,\pi_n^*(s^b)}\}} Q_n({s^a}') \nonumber \\ 
	& \ \ \ - \Theta_{n+1}(s_r,\pi_n^*(s_r)).
\end{align}
Similarly, for state $s^b$ we can get $Q_{n+1}(s^b) =  c (s^b,\pi_n^*(s^b)) + \sum\nolimits_{{s^b}' \in \mathcal{S}} {\Pr \{{s^b}'| {s^b,\pi_n^*(s^b)}\}} Q_n({s^b}') - \Theta_{n+1}(s_r,\pi_n^*(s_r))$. Note that for each sensor of state $s^b$ there exist two possible actions: $\pi_{i,n}^*(s^b) = 0$ and $\pi_{i,n}^*(s^b) = 1$. Regardless of which case, both $\tilde\Delta_i^{L,b}  \geq  \tilde\Delta_i^{L,a}$ and $ \tilde \Delta_i^{R,b} \geq \tilde \Delta_i^{R,a}$ hold, leading us to conclude that 
\begin{align}
	& \sum\limits_{{s^a}' \in \mathcal{S}} {\Pr \{{s^a}'| {s^a,\pi_n^*(s^b)} \}} Q_n({s^a}') \nonumber \\
	\leq & \sum\limits_{{s^b}' \in \mathcal{S}} {\Pr \{{s^b}'| {s^b,\pi_n^*(s^b)} \}} Q_n({s^b}'),
\end{align}
which can deduce inequality $Q_{n+1}(s^a) \leq Q_{n+1}(s^b)$ for $n >1$ via induction. Furthermore, due to the convergence of $Q_n(s)$, we can prove the inference of Lemma \ref{lemma2}.

\section{Proof of Theorem 2} \label{Proof Theorem 2}
In order to prove the existence of distributed structure, we first introduce the randomized policy $\pi_R$. Under the randomized policy $\pi_R$, the $i$-th sensor is scheduled with a certain probability $p_i^R$, such that the Bellman equation similar to \eqref{eq26} is established as follows:
\begin{align} 
	\phi^\dag + Q^R(s) \! = \! {c^\dag (s,d) +  {\sum\limits_{s' \in \mathcal{S}} {\mathbb{E}^{\pi_R}[\Pr \{s'| {s,d} \}}] Q^R(s')}},  \label{42}
\end{align}
where $\phi^\dag$ is the cost value under policy $\pi_R$ and $Q^R(s)$ is the corresponding value function. Then, the decomposition property of $Q^R(s)$ will be proved. According to the result in Theorem \ref{the1}, the policy $\pi^R$ is an unichain policy, such that the solution of \eqref{42} also exists. Depending on the the relationship between the joint distribution and the marginal distribution, for any state $s \in \mathcal{S}$ we have 
\begin{align}
	& \sum \limits_{s' \in \mathcal{S}} \Pr \{s'\left| {s,d}\right.\} = \sum \limits_{s' \in \mathcal{S}} \Pr \{(s^\Delta)', \vartheta' | s^\Delta, \vartheta, d\} \nonumber \\
	= & \sum \limits_{s_i' \in \mathcal{S}_i} \Pr \{(s_i^\Delta)', \vartheta' | s^\Delta, \vartheta, d \} \nonumber \\
	= & \sum \limits_{s_i' \in \mathcal{S}_i} \Pr \{(s_i^\Delta)', \vartheta' | s_i^\Delta, \vartheta, d_i \},
\end{align}
where $s_i$ is defined as $(s_i^\Delta, \vartheta)$ and its corresponding state space is $\mathcal{S}_i$. Hence, it is obvious that $\sum \nolimits_{s' \in \mathcal{S}} \Pr \{s'\left| {s,d}\right.\} = \sum \nolimits_{s_i' \in \mathcal{S}_i} \Pr \{s'\left| {s_i, d_i}\right.\}$. Combined with $\phi^\dag = \sum \nolimits_{i \in \mathcal{N}} \phi_i^\dag$ and $Q^R(s) = \sum \nolimits_{i \in N} Q^R_i(s_i)$, there exists following Bellman equation for the $i$-th sensor:
\begin{align} \label{40}
	\phi_i^\dag + Q_i^R(s_i)  = & {c_i^\dag(s_i,d_i)} \nonumber \\
	& + {\sum\limits_{s_i' \in \mathcal{S}_i} {\mathbb{E}^{\pi_R}[\Pr \{s_i'\left| {s_i,d_i} \right.}\}] Q_i^R(s_i')},
\end{align}
where $\Pr \{s_i'\left| {s_i,d_i}\right.\}$ can be computed from \eqref{eq11}, $c_i^\dag(s_i,d_i) = \mathbb{E} [ c_i (s_i,d_i)] = f (\Delta_i^R)$, $\phi_i^\dag$ denotes the optimal cost of the $i$-th sensor and $Q_i^R(s_i)$ is the corresponding value function. Then, by the concept of approximate dynamic programming \cite{ying_cui_delay-aware_2012}, the value function given in \eqref{eq26} is approximated to $Q^R(s)$ that equal to $\sum \nolimits_{i \in \mathcal{N}} Q^R_i(s_i)$, and we can further obtain the form of structure-informed scheduling policy $\tilde \pi^*$. 

Next, we will prove that the long-term time-average cost under $\tilde \pi^*$ is lower than under randomized policy $\pi^R$. Applying the proof of \cite[Proposition 4.4.2]{bertsekas_2007}, we express the performance of various policies as the gain-bias form. As policy $\pi^R$ is an unichain policy, define the gain of $\pi^R$ as a scalar $\psi^R$, and its bias is denoted as a vector $\boldsymbol{\nu}^R$ for long-term. Then, the simplified form of Bellman equation is written as follows:
\begin{align} \label{eq31}
	\psi^R \boldsymbol{1} + \boldsymbol{\nu}^R = \mathcal{T}^R \boldsymbol{\nu}^R,
\end{align}
where $\mathcal{T}^R$ is the Bellman operator and $\boldsymbol{1}$ denotes a vector where all elements are $1$. For the proposed structure-informed policy $\tilde \pi^*$, the unichain cannot be guaranteed, such that the gain is defined as a vector $\boldsymbol{\psi}^\dag$ and the corresponding bias is $\boldsymbol{\nu}^\dag$, the Bellman equation is given by $\boldsymbol{\psi}^\dag + \boldsymbol{\nu}^\dag = \mathcal{T}^\dag \boldsymbol{\nu}^\dag$, where $\mathcal{T}^\dag$ is the Bellman operator of $\tilde \pi^*$ with one-step cost vector $\boldsymbol{c}^\dag$. Due to the invariance of $\boldsymbol{\psi}^\dag$ under state transitions in long-term, we have $\boldsymbol{\psi}^\dag = P^\dag \boldsymbol{\psi}^\dag$, where $P^\dag$ is the transition probability matrix. 

In order to show the performance difference after applying policy $\tilde \pi^*$, we denote it by $\mu = \psi^R \boldsymbol{1} + \boldsymbol{\nu}^R - \mathcal{T}^\dag \boldsymbol{\nu}^R$. Then, substituting \eqref{eq31} and $\mathcal{T}^\dag$ into $\mu$ we have 
\begin{align} \label{eq33}
	\mu & = \psi^R \boldsymbol{1} + \boldsymbol{\nu}^R - (\boldsymbol{c}^\dag + P^\dag \boldsymbol{\nu}^R) \nonumber \\
	& = \psi^R \boldsymbol{1} - \boldsymbol{\psi}^\dag + \boldsymbol{\nu}^R - \boldsymbol{\nu}^\dag + \boldsymbol{\psi}^\dag + \boldsymbol{\nu}^R - (\boldsymbol{c}^\dag + P^\dag \boldsymbol{\nu}^R) \nonumber \\
	& = \psi^R \boldsymbol{1} - \boldsymbol{\psi}^\dag + \Upsilon - P^\dag \Upsilon.
\end{align}
By multiplying both sides of \eqref{eq33} by $P^\dag$ and iteratively summing $L$ steps, we can obtain 
{\small
	\begin{align}
		&   \ \ \ \lim \limits_{L \to \infty} \frac{1}{L} \sum \limits_{l=0}^{L-1} (P^\dag)^l \mu \nonumber \\
		& = \lim \limits_{L \to \infty} \frac{1}{L} \sum \limits_{l=0}^{L-1} (P^\dag)^l (\psi^R \boldsymbol{1} - \boldsymbol{\psi}^\dag + \Upsilon - P^\dag \Upsilon) \nonumber \\
		& = \lim \limits_{L \to \infty} \frac{1}{L} \sum \limits_{l=0}^{L-1} (P^\dag)^l (\psi^R \boldsymbol{1} - \boldsymbol{\psi}^\dag) \nonumber \\
		& \ \ \ + \lim \limits_{L \to \infty} \frac{1}{L} \sum \limits_{l=0}^{L-1} (P^\dag)^l (I - P^\dag) \Upsilon \nonumber \\
		& = \lim \limits_{L \to \infty} \frac{1}{L} [L (\psi^R \boldsymbol{1} - \boldsymbol{\psi}^\dag) + (I - (P^\dag)^L) \Upsilon] \nonumber \\
		& = \psi^R \boldsymbol{1} - \boldsymbol{\psi}^\dag.
\end{align}}

\noindent Then, it is easy to get $\lim \nolimits_{L \to \infty} \frac{1}{L} \sum \nolimits_{l=0}^{L-1} (P^\dag)^l \mu = \bar P^\dag \mu = \psi^R \boldsymbol{1} - \boldsymbol{\psi}^\dag$, where $\bar P^\dag$ is the convergence of $P^\dag$. Recall the MDP definition \cite[Proposition 5.4.2]{bertsekas_2007}, there must exists $\Pr\{s'\left| {s,d^a}\right.\}  \ne \Pr\{s'\left| {s,d^b}\right.\}$ under the assumption of $d^a \ne d^b$, resulting in $\mu > 0$. Hence, $\psi^R \boldsymbol{1} > \boldsymbol{\psi}^\dag$, thus the proposed scheduling policy performs better than the randomized policy in theory. The proof is accomplished.

\section{Proof of Theorem 3} \label{Proof Theorem 3}
To verify the threshold structure, we should first demonstrate the monotonicity of the decomposed value function $Q_i^R(s_i)$ based on Lemma \ref{lemma2}. Similiar to \eqref{35}, under randomized policy $\pi^R$, the single action-value function of the $n$th step iteration can be described as $\Theta_n^\dag(s_i) = {c_i (s_i,d_i)} + {\sum\nolimits_{s_i' \in \mathcal{S}_i} {\mathbb{E}^{\pi_R}[\Pr (s_i'\left| {s_i,d_i} \right.)}] Q_{i,n}^R(s_i')}$. $Q_{i,n}^R(s_i)$, like \eqref{36}, is iterated by $\Theta_n^\dag(s_i) - \Theta_n^\dag(s_i^r)$ with the definite state $s_i^r$. For the supposed states $s_i^a$ and $s_i^b$ with $\Delta_i^{L,a} \leq \Delta_i^{L,b}$, $\Delta_i^{R,a} \leq \Delta_i^{R,b}$ under a certain channel state, one has 
\begin{equation}
	Q_{i,n}^R(s_i^a) \leq Q_{i,n}^R(s_i^b),
\end{equation}
which can further deduce the monotonicity of $Q_i^R(s_i)$ according to the convergence of $Q_{i,n}^R(s_i)$.

Next, the submodularity of $\Theta^\dag (s,d)$ will be proved to show the threshold property, such that for $s^a, s^b$ with $\tilde \Delta^{L,a} = \tilde \Delta^{L,b}$ there exists $d_i =1$ when $\Delta_i^{R,a} \geq \Delta_i^{R,b}$. Besides, $\Theta^\dag (s^a,d) - \Theta^\dag (s^a,d') \leq \Theta^\dag (s^b,d) - \Theta^\dag (s^b,d')$ holds for $\forall i \in N$. By the definition of $\Theta^\dag (s,d)$, for any $s^a, s^b$ we have
\begin{align} \label{46}
	& \Theta^\dag (s^a,d) - \Theta^\dag (s^a,d') - \Theta^\dag (s^b,d) + \Theta^\dag (s^b,d') \nonumber \\
	\mathop  = \limits^{(a)} &  \sum \limits_{s^a \in \mathcal{S}} \Pr \{s_d^a|s^a,d\} Q^R(s_d^a) - \sum \limits_{s^a \in \mathcal{S}} \Pr \{s_{d'}^a|s^a,d'\} Q^R(s_{d'}^a) \nonumber \\
	& - \sum \limits_{s^b \in \mathcal{S}} \Pr \{s_d^b|s^b,d\} Q^R(s_d^b) + \sum \limits_{s^b \in \mathcal{S}} \Pr \{s_{d'}^b|s^b,d'\} Q^R(s_{d'}^b) \nonumber \\
	\mathop  = \limits^{(b)} & \sum \limits_{s^b \in \mathcal{S}} \Pr \{s_{d'}^b|s^b,d'\} Q^R(s_{d'}^b) \nonumber \\
	& - \sum \limits_{s^a \in \mathcal{S}} \Pr \{s_{d'}^a|s^a,d'\} Q^R(s_{d'}^a),
\end{align}
where the only difference of $s^a$ and $s^b$ is the AoFI that corresponding sensor is scheduled, i.e. $\Delta_i^{R,a} \geq \Delta_i^{R,b}$ for $d_i = 1$. Hence, the same states will be reached from $s^a$ and $s^b$ under the action $d$, by which $(a)$ is deduced to $(b)$ in \eqref{46}. For two terms in $(b)$, considering the situation of sensor $i$, $i \in \mathcal{N}^+ \triangleq \{ i|{d_i} = 1,i \in N\}$, because $s_{d'}^a = s_{d'}^b$ for other sensors. Then, for the case of $d_i' = 0$. there exists $s_{i,d'}^a = (\Delta_i^{L,a}, \Delta_i^{R,a} +1, \vartheta)$, $s_{i,d'}^b = (\Delta_i^{L,b}, \Delta_i^{R,b} +1, \vartheta)$. It is obvious that $\tilde \Delta^{L,a} \geq \tilde \Delta^{L,b}$, $\tilde \Delta^{R,a} \geq \tilde \Delta^{R,b}$ under this case, such that the RHS of \eqref{46} is negative and the following inequality holds:
\begin{equation}
	\Theta^\dag (s^a,d) - \Theta^\dag (s^a,d') \leq \Theta^\dag (s^b,d) - \Theta^\dag (s^b,d'). \label{47}
\end{equation}

From \eqref{44}, it is known that $\Theta^\dag(s, d) \leq \Theta^\dag(s, d')$ for $d \ne d'$, so $\tilde \pi^*(s) = d$. Assume that for state $s'$ there exists ${\Delta_i^R}^{\prime} \geq \Delta_i^R$ with $\tilde \Delta^L = {{}\tilde \Delta^L}^{\prime}$, and $\Theta^\dag(s', d) \leq \Theta^\dag(s', d')$ holds for $d \in \mathcal{A}$, $d \ne d'$ if Theorem \ref{the3} is right. According to the result of \eqref{47}, we can obtain $\Theta^\dag(s', d) - \Theta^\dag(s', d') \leq \Theta^\dag(s, d) - \Theta^\dag(s, d')$, thus $\tilde \pi^*(s) = \tilde \pi^*(s') = d$. The proof is completed accordingly. 

\section{Proof of Theorem 4} \label{Proof Theorem 4}
Define $\xi_{s_i}$ as the limiting distribution of state $s_i$, which is constrained by the proposed scheduling policy. According to the structured property of policy $\tilde \pi^*$, $d_i=1$ when AoRI $\Delta_{i}^R \geq {\Delta_{i,\vartheta}^R}^*$ under a certain channel state $\vartheta$. Define the initial state as $s_{i,0} = (0, 1, 0)$. As the Markov chain $s_i$ is irreducible, aperiodic and recurrent for policy $\tilde \pi^*$, such that any state has a positive probability to reach all states and the expectation time of the first time revisiting state $s_{i,0}$ is bounded, leading to the existence of unique limiting distribution of state $s_i$. Denote by ${{}\bar \Delta_i^R}^*$ the maximum AoRI under $\vartheta = 0$ and $\vartheta = 1$ as $\max \{ {\Delta_{i,\vartheta}^R}^* | {\vartheta \in \{0,1\}} \}$, and then we consider the following cases for limiting distribution $\xi_{s_i}$:

\begin{enumerate}
	\item For the case of $0 < \Delta_i^R \leq {{}\bar \Delta_i^R}^*$, if ${\Delta_i^R}' = \Delta_i^R + 1$, then we have
	\begin{equation}
		\xi_{s_i'} = \sum \limits_{\vartheta \in \{0,1\}} \xi_{s_i} \Pr \{ {s_i'} | {{s_i},{\tilde \pi ^ * }} \}. \label{eq36}
	\end{equation}
	When $\Delta_i^R \leq {\Delta_{i,\vartheta}^R}^*$, then the value of $\Pr \{{s_i'}\left| {{s_i},{\tilde \pi ^ *}} \right.\} $ is $\lambda_i \kappa_{\vartheta \vartheta'}$ for ${\Delta_i^L}^{\prime} = 0$ or $(1 - \lambda_i) \kappa_{\vartheta \vartheta'}$ for ${\Delta_i^L}^{\prime} = \Delta_i^L +1$. When $\Delta_i^R > {\Delta_{i,\vartheta}^R}^*$, the corresponding situation is no data arrival or transmission failure, then the value of $\Pr ({s_i'}| {s_i},{\tilde \pi ^ *})$ is given as $\kappa_{\vartheta \vartheta'} (1 - \lambda_i p_i^{\vartheta'}(k))$. \label{item1}
	\item For the case of $\Delta_i^R > {{}\bar \Delta_i^R}^*$, the value of $\Pr \{{s_i'}|{s_i},{\tilde \pi^*}\}$ can refer to the second situation in case \ref{item1}, such that $\Pr ({s_i'}| {s_i},{\tilde \pi^*}) = \kappa_{\vartheta \vartheta'} (1 - \lambda_i p_i^{\vartheta'}(k))$. \label{item2}
\end{enumerate}

It is obvious that $\mathcal{C}(s_0,\pi)$ is bounded under case \ref{item1}) due to $\Delta_i^R \leq {{}\bar \Delta_i^R}^*$, but case \ref{item2}) with $\Delta_i^R > {{}\bar \Delta_i^R}^*$ cannot guarantee the cost boundness. We expand $\xi_{s_i'}$ of case \ref{item2}) based on two different channel states and obtain the following form:
\begin{align} \label{eq35}
	\xi_{s_i'(\vartheta' = 0)} & = \xi_{s_i(\vartheta = 0)}\kappa _{00}(1 - \lambda_i{p_i^0}(k)) \nonumber \\
	& \ \ \ + \xi_{s_i(\vartheta = 1)}\kappa _{10}(1 - \lambda_i{p_i^0}(k)), \nonumber \\
	\xi_{s_i'(\vartheta' = 1)} & = \xi_{s_i(\vartheta = 0)}\kappa _{01}(1 - \lambda_i{p_i^1}(k)) \nonumber \\
	& \ \ \ + \xi_{s_i(\vartheta = 1)}\kappa _{11}(1 - \lambda_i{p_i^1}(k)).
\end{align}

In order to simplify the presentation, we define $\boldsymbol{\xi}_{s_i} = (\xi_{s_i(\vartheta = 0)},\xi_{s_i(\vartheta = 1)})$. According to the form of \eqref{eq35}, it can be further expressed as $\boldsymbol{\xi}_{s_i'} = \boldsymbol{\xi}_{s_i} \Omega (I - \lambda_i \mathcal{P}_i)$. After continuous iterations, we get when AoRI evolves to $({{}\bar \Delta_i^R}^* + \tau )$, $\tau \in \mathbb{N}^+$,
\begin{align}\label{eq39}
	\boldsymbol{\xi}_{s_i({{}\bar \Delta_i^R}^* + \tau)} = \boldsymbol{\xi}_{s_i({{}\bar \Delta_i^R}^*)} (\Omega (I - \lambda_i \mathcal{P}_i))^\tau. 
\end{align}

According to the AoI function proposed in \eqref{3.8}, the convergence of time-average monitoring cost $\mathcal{C}(s_0,\pi)$ is directly related to $h^{\Delta_i^R}(\bar P_i)$. Hence, we have the following equation based on the limiting distribution concept: 
\begin{align} \label{eq37}
	\mathcal{E} = \sum \limits_{i=1}^N \sum \limits_{\Delta_i^R=1}^\infty  \tilde \xi_{s_i(\Delta_i^R)} h^{\Delta_i^R}(\bar P_i),
\end{align}
where $h(\bar P_i) \triangleq {A_i} \bar P_i A_i^T + \Sigma _i^\omega $, $\tilde \xi_{s_i(\Delta_i^R)} \triangleq \sum \nolimits_{\vartheta \in \{0,1\}} \xi_{s_i(\Delta_i^R, \vartheta)} = || {\boldsymbol{\xi}_{s_i(\Delta_i^R)}}||_1$ and $\sum \nolimits_{\Delta_i^R=1}^\infty \tilde \xi_{s_i(\Delta_i^R)} = 1$. By the definition of $h(\bar P_i)$, $h^{\Delta_i^R}(\bar P_i)$ is monotone non-decreasing such that $E < \infty$ can ensure the boundedness of time-average cost, so as to the stability of monitoring performance. Using the matrix theory, the eigenvalue of nonnegative matrix $\Omega (I - \lambda_i \mathcal{P}_i)$ as $\rho(\Omega (I - \lambda_i \mathcal{P}_i))$, thus there exists a eigenvector $\boldsymbol{\varepsilon}_i$ satisfying $\rho(\Omega (I - \lambda_i \mathcal{P}_i))\boldsymbol{\varepsilon}_i = (\Omega (I - \lambda_i \mathcal{P}_i))\boldsymbol{\varepsilon}_i$. Then, it is easy to get $\boldsymbol{\xi}_{s_i({{}\bar \Delta_i^R}^* + \tau)} \boldsymbol{\varepsilon}_i = \boldsymbol{\xi}_{s_i({{}\bar \Delta_i^R}^*)} \rho^\tau (\Omega (I - \lambda_i \mathcal{P}_i)) \boldsymbol{\varepsilon}_i$. There must exist a constant $e_\tau$ that $||\boldsymbol{\xi}_{s_i({{}\bar \Delta_i^R}^* + \tau)}||_1 = e_\tau \rho^\tau (\Omega (I - \lambda_i \mathcal{P}_i)) ||\boldsymbol{\xi}_{s_i({{}\bar \Delta_i^R}^*)}||_1$. Set $\zeta_i = \sum \nolimits_{\Delta_i^R = {{}\bar \Delta_i^R}^*}^\infty \rho^{\Delta_i^R - {{}\bar \Delta_i^R}^*}(\Omega (I - \lambda_i \mathcal{P}_i)) ||\boldsymbol{\xi}_{s_i({{}\bar \Delta_i^R}^*)}||_1 h^{\Delta_i^R}(\bar P_i)$, we obtain the upper boundness of \eqref{eq37} as $\mathcal{E} \leq \sum \nolimits_{i=1}^N \bar e_\tau \zeta_i$
\begin{equation} 
	\mathcal{E} \leq \sum \limits_{i=1}^N \bar e_\tau \zeta_i, \label{eq41}
\end{equation}
where $\bar e_\tau = \max_{\tau \in \mathbb{N}}\{e_\tau\}$ can guarantee the validity of \eqref{eq41}. With the whole form of $h(\bar P_i)$, $\zeta_i$ is expended as follows:
	\begin{align} \label{eq43}
		\zeta_i & = \sum \limits_{j = {{}\bar \Delta_i^R}^*}^\infty \rho^{j - {{}\bar \Delta_i^R}^*}(\Omega (I - \lambda_i \mathcal{P}_i)) ||\boldsymbol{\xi}_{s_i({{}\bar \Delta_i^R}^*)}||_1 h^{j}(\bar P_i) \nonumber \\
		& = \sum \limits_{j = {{}\bar \Delta_i^R}^*}^\infty \rho^{j - {{}\bar \Delta_i^R}^*}(\Omega (I - \lambda_i \mathcal{P}_i)) ||\boldsymbol{\xi}_{s_i({{}\bar \Delta_i^R}^*)}||_1 \nonumber \\
		& \ \ \ \times \left( {A_i^j \bar P_i (A_i^T)^j} + \sum \limits_{l=0}^{j-1} A_i^l \Sigma _i^\omega  (A_i^T)^l \right) \nonumber \\   
		& = ||\boldsymbol{\xi}_{s_i({{}\bar \Delta_i^R}^*)}||_1 A_i^{{{}\bar \Delta_i^R}^*} \left( {\sum \limits_{l=0}^\infty \varsigma_{1,i} (l) } \right) (A_i^T)^{{{}\bar \Delta_i^R}^*} \nonumber \\
		& \ \ \ + ||\boldsymbol{\xi}_{s_i({{}\bar \Delta_i^R}^*)}||_1 \frac{1}{\varsigma_{3,i}} \sum \limits_{l={\Delta_i^R}^* - 1}^\infty \varsigma_{2,i} (l),
\end{align}

\noindent where $\varsigma_{1,i} (l) = \rho^l(\Omega (I - \lambda_i \mathcal{P}_i)) A_i^l \bar P_i (A_i^T)^l$, $\varsigma_{2,i} (l) = \rho^l(\Omega (I - \lambda_i \mathcal{P}_i)) A_i^l \Sigma _i^\omega(A_i^T)^l$ and $\varsigma_{3,i} = \rho^{{{}\bar \Delta_i^F}^* - 1}(\Omega (I - \lambda_i \mathcal{P}_i))$. By expanding \cite[Lemma 3]{ren_dynamic_2014}, $\zeta_i$ is surely bounded if and only if $\rho (\Omega (I - \lambda_i \mathcal{P}_i)) < \frac{1}{\rho^2(A_i)}$ holds. Therefore, $\mathcal{E} < \infty$ is ensured and the monitoring cost is convergent.

\section{Proof of Corollary 2} \label{Proof Corollary 2}
Considering that the data arrival process follows a Markov chain, the occurrence of current data arrival depends on the arrival status in the previous time slot. Therefore, we can obtain the state transition of $\Delta_i^L$ as follows:
	\begin{align} \label{eq44}
		& \Pr \{\gamma_i^R (k) | \gamma_i^R (k-1)\} \nonumber \\
		& = \begin{cases}
			\tilde \lambda_i,     & \text{if} \ \gamma_i^R (k) = 0, \gamma_i^R (k-1) = 0, \\
			1 - \tilde \lambda_i, & \text{if} \ \gamma_i^R (k) = 1, \gamma_i^R (k-1) = 0, \\
			\bar \lambda_i,       & \text{if} \ \gamma_i^R (k) = 1, \gamma_i^R (k-1) = 1, \\
			1 - \bar \lambda_i,   & \text{if} \ \gamma_i^R (k) = 0, \gamma_i^R (k-1) = 1, \\
			0,                    & \text{otherwise}.
		\end{cases}
\end{align}

\noindent Then, the state transition of $s_i^\Delta (k)$ defined in \eqref{16} is transformed as the following form:

1) For $d_i(k) = 0$, there exists
	\begin{align}
		& \Pr \{ {s_i^\Delta(k) | {s_i^\Delta(k-1), \vartheta(k-1), d_i(k) = 0} } \} \nonumber \\
		& =\begin{cases}
			\tilde \lambda_i,                             & {\text{if}} \ \Delta _i^L (k-1) > 0, \\
			& s_i^\Delta(k) = ( \Delta _i^L (k-1) + 1, {\Delta _i^R}(k-1) + 1 ), \\
			\bar \lambda_i,                               & {\text{if}} \ \Delta _i^L (k - 1) = 0, \\
			& s_i^\Delta(k) = (0, {\Delta _i^R}(k-1) + 1), \\
			1 - \tilde \lambda_i,                         & {\text{if}} \ \Delta _i^L (k-1) > 0, \\
			& s_i^\Delta(k) = ( 0, {\Delta _i^R}(k-1) + 1 ), \\
			1 - \bar \lambda_i,                           & {\text{if}} \ \Delta _i^L (k-1) = 0, \\
			& s_i^\Delta(k) = (1, {\Delta _i^R}(k-1) + 1),  \\
			0,                                            & {\text{otherwise}}.
		\end{cases}
\end{align}

2) For $d_i(k) = 1$, there exists
	\begin{align}
		& \Pr \{ {s_i^\Delta(k) | {s_i^\Delta(k-1), \vartheta(k-1), d_i(k) = 1}} \} \nonumber \\
		& =\begin{cases}
			\tilde \lambda_i p_i^\vartheta (k),           & {\text{if}} \ \Delta _i^L (k-1) > 0,  \\
			& s_i^\Delta(k) =(\Delta _i^L (k-1) + 1,\\
			& {\Delta _i^L}(k-1) + 1), \\
			\bar \lambda_i p_i^\vartheta (k),             & {\text{if}} \ \Delta _i^L (k-1) = 0,  \\
			& s_i^\Delta(k) = (0,1), \\ 
			(1 - \tilde \lambda_i) p_i^\vartheta (k),     & {\text{if}} \ \Delta _i^L (k-1) > 0, \\
			& s_i^\Delta(k) = (0, {\Delta _i^L}(k-1) + 1), \\
			(1 - \bar \lambda_i) p_i^\vartheta (k),       & {\text{if}} \ \Delta _i^L (k-1) = 0, \\
			& s_i^\Delta(k) = (1,1), \\
			\tilde \lambda_i (1 - p_i^\vartheta(k)),      & {\text{if}} \ \Delta _i^L (k-1) > 0,  \\
			& s_i^\Delta(k) = (\Delta _i^L (k-1) + 1,  \\
			& {\Delta _i^R}(k-1) + 1), \\
			\bar \lambda_i (1 - p_i^\vartheta(k)),        & {\text{if}} \ \Delta _i^L (k-1) = 0,  \\
			& s_i^\Delta(k) = (0,1), \\
			(1 - \tilde \lambda_i) (1 - p_i^\vartheta(k)),& {\text{if}} \ \Delta _i^L (k-1) > 0, \\
			& s_i^\Delta(k) = (0, {\Delta _i^R}(k-1) + 1), \\
			(1 - \bar \lambda_i) (1 - p_i^\vartheta(k)),  & {\text{if}} \ \Delta _i^L (k-1) = 0, \\
			& s_i^\Delta(k) = (1, {\Delta _i^R}(k-1) + 1), \\
			0,                                            & {\text{otherwise}}.
		\end{cases}
\end{align}

\noindent Referring to the proof of Theorem \ref{thm4}, the boundness of $\mathcal{C}(s_0, \pi)$ cannot be surely bounded when $\Delta_i^R > {{}\bar \Delta_i^R}^*$. Then, similar to \eqref{eq35}, we can obtain the distribution of $\xi_{s_i'}$ under different arrival situations. According to \eqref{eq44}, we have following cases:
	\begin{align} \label{eq46}
		& \xi_{s_i'(\vartheta' = 0)} \nonumber \\ & = 
		\begin{cases}
			\xi_{s_i(\vartheta = 0)}\kappa _{00}(1 - (1 - \tilde \lambda_i){p_i^0}(k))       \\
			\ + \xi_{s_i(\vartheta = 1)}\kappa _{10}(1 - (1 - \tilde \lambda_i){p_i^0}(k)), &  \text{if} \ \gamma_i^R (k-1) = 0, \\
			\xi_{s_i(\vartheta = 0)}\kappa _{00}(1 - \bar \lambda_i{p_i^0}(k))               \\
			\ + \xi_{s_i(\vartheta = 1)}\kappa _{10}(1 - \bar \lambda_i{p_i^0}(k)),   &  \text{if} \ \gamma_i^R (k-1) = 1. \\
		\end{cases} 
\end{align}

\begin{align} \label{eq47}
		& \xi_{s_i'(\vartheta' = 1)} \nonumber \\ & = 
		\begin{cases}
			\xi_{s_i(\vartheta = 0)}\kappa _{01}(1 - (1 - \tilde \lambda_i){p_i^1}(k))       \\
			\ + \xi_{s_i(\vartheta = 1)}\kappa _{11}(1 - (1 - \tilde \lambda_i){p_i^1}(k)), &  \text{if} \ \gamma_i^R (k-1) = 0,  \\
			\xi_{s_i(\vartheta = 0)}\kappa _{01}(1 - (1 - \bar \lambda_i){p_i^1}(k))         \\
			\ + \xi_{s_i(\vartheta = 1)}\kappa _{11}(1 - (1 - \bar \lambda_i){p_i^1}(k)),   &  \text{if} \ \gamma_i^R (k-1) = 1.  \\
		\end{cases}
\end{align}

\noindent Due to the fact that $\gamma_i^R (k-1)$ has two possible states, $\xi_{s_i'(\vartheta' = 0)}$ and $\xi_{s_i'(\vartheta' = 1)}$ can be discussed separately. Based on the forms of \eqref{eq46} and \eqref{eq47}, the original expression in \eqref{eq39} is thus be derived in two corresponding cases as follows:
\begin{align}
		& \boldsymbol{\xi}_{s_i({{}\bar \Delta_i^R}^* + \tau)} \nonumber \\ & =  
		\begin{cases}
			\boldsymbol{\xi}_{s_i({{}\bar \Delta_i^R}^*)} (\Omega (I - (1 - \tilde \lambda_i) \mathcal{P}_i))^\tau,  &  \text{if} \ \gamma_i^R (k-1) = 0,\\
			\boldsymbol{\xi}_{s_i({{}\bar \Delta_i^R}^*)} (\Omega (I -  \tilde \lambda_i \mathcal{P}_i))^\tau,       &  \text{if} \ \gamma_i^R (k-1) = 1.\\
		\end{cases}
\end{align}

\noindent And the remaining proof steps follow the same procedure as in \eqref{eq37}-\eqref{eq43}. To ensure the convergence of the monitoring cost, both $\rho (\Omega (I - (1 - \tilde \lambda_i) \mathcal{P}_i)) < \frac{1}{\rho^2 (A_i)}$ and $\rho (\Omega (I -  \bar \lambda_i \mathcal{P}_i)) < \frac{1}{\rho^2 (A_i)}$ should be satisfied, such that we can get the stability condition for the Markov arrival process as $\rho (\Omega (I -  \hat \lambda_i \mathcal{P}_i)) < \frac{1}{\rho^2 (A_i)}$, where $\hat \lambda_i = \min \{1 - \tilde \lambda_i, \bar \lambda_i\}$. The proof is completed.

\ifCLASSOPTIONcaptionsoff
  \newpage
\fi

\bibliographystyle{IEEEtran}
\bibliography{reference}

\end{document}